\documentclass[12pt]{article}
\usepackage{epsfig, amssymb, graphicx, amsmath} 
\usepackage{amsthm}
\usepackage{relsize}

\usepackage[round-mode=places, round-integer-to-decimal, round-precision=2,
table-format = 1.2, 
table-number-alignment=center,
round-integer-to-decimal,
output-decimal-marker={,}
]{siunitx} 
\usepackage{booktabs}
\usepackage{enumitem}
\usepackage{eurosym}	
\usepackage{wrapfig}	
\usepackage{float}		
\usepackage{subfig}	
\usepackage{footnote}
\usepackage{thmtools}
\usepackage{bbm}		
\usepackage[title]{appendix}
\usepackage{adjustbox}

\usepackage[authoryear]{natbib}
\bibpunct{(}{)}{,}{a}{}{,}

\newtheorem{theorem}{Theorem}

\newtheorem{proposition}[theorem]{Proposition}

\newtheorem{lemma}[theorem]{Lemma}

\theoremstyle{definition}

\theoremstyle{remark}

\numberwithin{theorem}{section}
\numberwithin{proposition}{section}
\numberwithin{lemma}{section}
\numberwithin{corollary}{section}
\numberwithin{definition}{section}
\numberwithin{remark}{section}
\numberwithin{example}{section}

\newcommand{\be}{\begin{equation}}
\newcommand{\en}{\end{equation}}
\newcommand{\ben}{\begin{equation*}}
\newcommand{\enn}{\end{equation*}}
\newcommand{\bea}{\begin{eqnarray}}
\newcommand{\ena}{\end{eqnarray}}

\begin{document}
	
	\newlength\tindent
	\setlength{\tindent}{\parindent}
	\setlength{\parindent}{0pt}
	\renewcommand{\indent}{\hspace*{\tindent}}
	
	\begin{savenotes}
		\title{
			\bf{ 
				Additive normal tempered stable  processes \\
				for equity derivatives and power law scaling
		}}
		\author{
			Michele Azzone$^\ddagger$ \& 
			Roberto Baviera$^\ddagger$ 
		}
		
		\maketitle
		
		\vspace*{0.11truein}
		\begin{tabular}{ll}
			$(\ddagger)$ &  Politecnico di Milano, Department of Mathematics, 32 p.zza L. da Vinci, Milano \\
		\end{tabular}
	\end{savenotes}
	
	\vspace*{0.11truein}
	\begin{abstract}
		\noindent
		We introduce a simple additive process for equity index derivatives.
		The model generalizes L\'evy Normal Tempered Stable processes (e.g. NIG and VG) with time-dependent parameters. 
		It accurately fits the equity index volatility surfaces in the whole time range of quoted instruments, including options with small time-horizon (days) and long time-horizon (years).
		
		\noindent
		We introduce the model via its characteristic function. This allows using classical Fourier pricing techniques. 
		We discuss the calibration issues in detail and
		we show that,  in terms of mean squared error, calibration is on average two orders of magnitude better than both L\'evy  and Sato processes alternatives.
		
		\noindent
		We show that even if the model loses the classical stationarity property of L\'evy processes, it presents interesting scaling properties for the calibrated parameters.
	\end{abstract}
	
	\vspace*{0.11truein}
	\vspace*{0.11truein}
	
	\vspace*{1cm}
	\begin{flushleft}
		Cite as: Azzone, A. \& Baviera, R. (2021). 	Additive normal tempered stable  processes \\
		for equity derivatives and power law scaling. \textit{Quantitative Finance}, Published online. \\
	\end{flushleft}
	\vspace{3cm}
	\begin{flushleft}
		{\bf Address for correspondence:}\\
		Roberto Baviera\\
		Department of Mathematics \\
		Politecnico di Milano\\
		32 p.zza Leonardo da Vinci \\ 
		I-20133 Milano, Italy \\
		Tel. +39-02-2399 4630\\
		Fax. +39-02-2399 4621\\
		roberto.baviera@polimi.it
	\end{flushleft}

	\newpage
	
	\begin{center}
		\Large\bfseries 
		Additive normal tempered stable processes \\
		for equity derivatives and power law scaling
	\end{center}
	
	\vspace*{0.21truein}
	
	\section{Introduction}
	
	Following the seminal work of \citet{MadanSeneta1990}, L\'evy processes have become a powerful modeling solution that provides parsimonious models consistent with option prices and with underlying asset prices. 
	There are several advantages of this modeling approach: 
	this model class admits a simple closed formula \citep{carr1999option, lewis2001simple} and it is a parsimonious description of some key features of the market volatility surface. 
	In particular, the class of L\'evy normal tempered stable processes (LTS) appears to be rather flexible and it involves very few parameters. LTS are pure jump\footnote{The relevance of pure jump dynamics in the equity and commodity asset classes has been discussed in the recent literature, see, e.g. \citet[][]{ornthanalai2014levy,li2014time,ballotta2018smiles}.} processes, obtained via the well-established L\'evy subordination technique  \citep[see, e.g.][]{Cont, Schoutens}.
	Specifically, most of these applications involve two processes in the LTS family: Normal Inverse Gaussian (NIG) \citep{barndorff1997normal} and Variance Gamma (VG) \citep{madan1998variance}, which are obtained via two different L\'evy subordinators. Both NIG and VG are characterized by three parameters: 
	$\sigma$, which controls the average level of the volatility surface;
	$k$, which is related to the convexity of the implied volatility surface;  and $\eta$, which is linked to the volatility \textit{skew} \citep[for a definition see, e.g.] [Ch.3, p.35]{gatheral2011volatility}. 
	
	\smallskip
	
	Unfortunately, the recent literature
	has shown that these models do not reproduce the implied volatilities that are observed in the market data at different time horizons with sufficient precision
	\citep[see, e.g.][Ch.14]{Cont}.
	L\'evy normal tempered stable processes are pure jump models with independent and stationary increments.
	The key question is as follows: is it reasonable to consider stationary increments when modeling implied volatility?
	Jump stationarity is a feature that significantly simplifies the model's characteristics but it is rather difficult to justify {\it a priori} from a financial point of view.
	For example, a {\it market maker} in the option market does not consider the consequences of a jump to be equivalent on options with different maturities.
	He cares about the amount of trading in the underlying required to replicate the option after a jump arrival.
	The impact of such a jump on the hedging policy is inhomogeneous with option maturity.\footnote{Gamma is the Greek measure that quantifies the amount of this hedging and, generally, it decreases with time-to-maturity.}
	Hence, {\it a priori},
	it is not probable that a stationary model can adequately describe implied volatilities. 
	
	\smallskip
	
	Additive processes have been proposed to overcome this problem. Additive processes are an extension of L\'evy processes that consider independent but not stationary increments. 
	Given an additive process, for every fixed time $t$, it is always possible to define a L\'evy process that at time $t$ has the same law as the additive process. 
	This feature allows us to maintain several properties (both analytical and numerical) of the L\'evy processes. 
	
	
	The probability description of additive processes is well-established \citep{Sato} but the applications in quantitative finance are relatively few. 
	A first application of additive processes to option pricing is developed by \citet{carr2007self}, 
	who investigate Sato processes \citep{sato1991self} in derivative modeling \citep[see also][]{eberlein2009sato}. 
	\citet{benth2012risk} use additive processes, which they call time-inhomogeneous Levy processes, in the electricity market. In their paper, the electricity spot price is characterized by Ornstein-Uhlenbeck processes, which are driven by additive processes.
	More recently, \citet{li2016Additive} have considered a larger class of additive processes. Their paper studies additive subordination, which (they show) is a useful technique for constructing time inhomogeneous Markov processes with an analytically tractable characteristic function. This technique is a natural generalization of L\'evy subordination.
	
	\bigskip
	We introduce a new class of (pure-jump) additive processes through their characteristic function which are named additive normal tempered stable (ATS) processes.
	ATS processes (in general) cannot be obtained via a time-change as in the additive subordination of \citet{li2016Additive} and are not  Sato processes. There is a subclass of ATS obtained via additive subordination and a subclass of Sato ATS process.
	The main advantage of this new class of models is the possibility to ``exactly" calibrate the term structure of observed implied volatility surfaces, while 
	maintaining the parsimony of LTS.

	We provide a calibration example of the ATS on the S\&P 500 and EURO STOXX 50 implied volatility surfaces of the $30^{th}$ of May 2013. The  ATS calibration is on average two orders of magnitude better than the corresponding LTS in terms of mean squared error. 
	We show that the calibrated time-dependent parameters present an interesting and statistically relevant self-similar behavior compatible with a power-law scaling subcase of ATS.
	Moreover, we have verified that these results are robust over time.
	
	\bigskip
	
	The main contributions of this paper are threefold. First,  we introduce a new broad family of additive processes, which we call additive normal tempered stable (ATS) processes.\\ 
	Second, we calibrate the ATS processes on S\&P 500 and EURO STOXX 50 volatility surfaces. We show that ATS  have better calibration features 
	(in terms of both the Mean Squared Error and the Mean Absolute Percentage Error) than LTS and Sato processes.\\
	Finally, we consider a re-scaled ATS process via a time-change based on the implied volatility term structure. 
	We show that the calibrated parameters exhibit a self-similar behavior w.r.t. the new time. The statistical relevance of these scaling properties is verified. 
	
	\smallskip
	
	The rest of the paper is organized as follows. 
	In Section {\bf \ref{section:Themodel}}, we introduce the model: we prove that there exists a  new family of additive processes as the natural extension of the corresponding L\'evy processes. 
	In Section {\bf \ref{section:calibration}},
	we describe the dataset used in the calibration, the calibration results for ATS, LTS, and Sato processes
	and an interesting scaling property of the calibrated parameters.
	In Section {\bf \ref{section:AdditionalResult}}, we show that LTS and Sato processes fail to reproduce some stylized facts observed in market data, which are adequately described by ATS processes and we present a robustness analysis.
	Finally,
	Section {\bf \ref{section:conclusions}} concludes.
	
	\section{The model}
	\label{section:Themodel}
	In this Section, we introduce the ATS process, 
	as a natural extension of the LTS process, that, on the one hand, maintains the increments' independence as in the corresponding Levy process, and, on the other hand, allows for time-inhomogeneous parameters. 
	First, we prove a sufficient condition for the existence of ATS processes (\textbf{Theorem \ref{theorem:f_Additive}}) and 
	the martingale property for the corresponding forward process (\textbf{Proposition \ref{th6}}).  
	Then, we introduce the power-law scaling ATS as a subcase of a generic ATS (\textbf{Theorem \ref{theorem:semplified_f}}): 
	we show in the next Section that this model describes accurately the implied volatility surface.
	Finally, we prove a key model feature: the model allows to reproduce 
	a generic volatility term structure (\textbf{Proposition \ref{theorem:NewAdditive}}).
	\bigskip
	
	L\'evy normal tempered stable processes (LTS) are commonly used in the financial industry for derivative pricing.
	According to this modeling approach,
	the forward  with expiry $T$  is an exponential L\'evy; i.e. 
	\be
	F_t (T) := F_0 (T) \; \exp(f_t) \;, 
	\label{eq:FwdBasic}
	\en
	with $f_t$ a LTS 
	\begin{equation*}
	f_t =  - \left(  \frac{1}{2}+\eta \right) \; \sigma^2 \; S_t +  \sigma \; W_{S_t}+ \varphi \,t\qquad \forall t \in [0,T] \; ,
	\end{equation*}
	where 
	$\eta, \sigma$ are two real parameters ($\eta \in \mathbb{R}, \sigma \in \mathbb{R}^+$), 
	while 
	$\varphi$ is obtained by imposing the martingale condition on 
	$F_t (T)$.\footnote{A parametrization scheme of the drift in terms of $\eta$ can be suitable for some applications: $\eta$ controls the volatility \textit{skew}.
		In particular, it can be proven that for $\eta =0$ the smile is symmetric, i.e. the implied volatility \textit{skew} is zero
		\citep[see, e.g.][Prop. p.21]{baviera2007}.
	}
	$W_t$ is a Brownian motion and $S_t$ is a L\'evy tempered stable subordinator independent from the Brownian motion with variance per unit of time $k$. Examples of LTS subordinators are the Inverse Gaussian process for NIG or the Gamma process for VG.

	It is possible to write the characteristic function of  $f_t$ as 
	\begin{equation}
	\label{laplace_levy}
	\mathbb{E}\left[e^{iuf_t}\right]={\cal L}_t \left(iu \left(\frac{1}{2}+\eta\right)\sigma^2+\frac{u^2\sigma^2}{2}; \;k,\;\alpha\right)e^{iu\varphi \, t}\;\;,
	\end{equation}
	where $\alpha\; \in\; [0,1)$ is the LTS \textit{index of stability} and ${\cal L}_t$ is the Laplace transform of $S_t$ \begin{equation}
 \ln {\cal L}_t \left(u;\;k,\;\alpha\right) :=
\begin{cases} 
	\displaystyle \frac{t}{k}
	\displaystyle \frac{1-\alpha}{\alpha}
	\left \{1-		\left(1+\frac{u \; k}{1-\alpha}\right)^\alpha \right \} & \mbox{if } \; 0< \alpha < 1 \\[4mm]
	\displaystyle -\frac{t}{k}
	\ln \left(1+u \; k\right)  & \mbox{if } \; \alpha = 0 \end{cases}\;\; . \label{eq:lap_transf}
	\end{equation}
	This theory is well known and can be found in many excellent textbooks  \citep[see, e.g.][]{Cont, Schoutens}.
	
	\bigskip
	
	As already discussed in the Introduction, LTS processes do not properly describe short and long maturities at the same time, 
	while they allow an excellent calibration for a fixed maturity.
	For this reason, we would like to select a process that allows independent but non-stationary increments: i.e. an additive process. The simplest way to obtain this modeling feature is to consider an additive process with a characteristic function of the same form of (\ref{laplace_levy}) but with time-dependent parameters
	
	\begin{equation}
	\label{laplace}
	\mathbb{E}\left[e^{iuf_t}\right]={\cal L}_t \left(iu \left(\frac{1}{2}+\eta_t \right)\sigma_t^2+\frac{u^2\sigma^2_t}{2};\;k_t,\;\alpha \right)e^{iu\varphi_tt}\;\;,
	\end{equation}
	where $ \sigma_t $, $k_t$ are continuous on $[0,\infty)$ and $ \eta_t $, $\varphi_t$ are continuous  on $(0,\infty)$
	with $ \sigma_t > 0$, $ k_t \geq  0$ and  $\varphi_t \,t$ goes to zero as $t$ goes to zero. $\alpha \in [0,1)$ as in the LTS case.
	
	In \textbf{Theorem \ref{theorem:f_Additive}}, we prove that this process exists if some conditions on $\sigma_t$, $\eta_t$ 
	and $k_t$ are satisfied.
	
	\bigskip

	An additive process is a c\'adl\'ag stochastic process on $\mathbb{R}$	
	$\left\{X_t\right\}_{t\geq 0}$, with $X_0=0$ a.s. and characterized 
	by independent increments and stochastic continuity
	\citep[see, e.g.][Def.14.1 p.455]{Cont}.
	It can be proven that the distribution of an additive process at time $t$ is infinitely divisible.  $(A_t,\nu_t,\gamma_t)$ is the generating triplet that characterizes the additive process $ \left\{{X_t}\right\}_{t\geq 0}$. $A_t$, $\nu_t$ and $\gamma_t$ are called respectively the diffusion term, the L\'evy measure and the drift term \citep[see][pp.38-39]{Sato}.\footnote{In this paper, the notation follows closely the one in \citet{Sato}.} 
	
	\bigskip
	

	\citet[][Th.9.8,  p.52]{Sato} proves a powerful link between a system of infinitely divisible probability distributions and the existence of an additive process. 
	In particular, Sato requires two main classes of conditions on the generating triplet: i) some conditions of monotonicity, necessary to avoid  {meaningless} negative diffusion term or negative L\'evy measure for the process increments and ii) some continuity conditions, in order to obtain the stochastic continuity.
	
	We use \citet[][Th.9.8, p.52]{Sato} to prove the main theoretical results in this paper: there exists a family of additive processes with characteristic function (\ref{laplace}).
	\begin{theorem} {\bf Sufficient conditions for existence of ATS} \label{theorem:f_Additive}\\ 
		There exists an additive process $ \left\{f_t \right\}_{t\geq 0}$ with the characteristic function  (\ref{laplace}) if the following two conditions hold.
		\begin{enumerate}
			\item $g_1(t)$, $g_2(t)$, and $g_3(t)$ are non decreasing, where  
			\begin{align*}
			g_1(t)&:=(1/2+\eta_t)-\sqrt{\left(1/2+\eta_t\right)^2+ 2(1-\alpha)/(\sigma_t^2 k_t)}\\
			g_2(t)&:=-(1/2+\eta_t)-\sqrt{\left(1/2+\eta_t\right)^2+ 2(1-\alpha)/( \sigma_t^2 k_t)}\\
			g_3(t)&:=\frac{ t^{1/\alpha}\sigma^2_t}{k_t^{(1-\alpha)/\alpha}}\sqrt{\left(1/2+\eta_t\right)^2+ 2(1-\alpha)/( \sigma_t^2k_t)}\;\;;
			\end{align*}
			\item Both  $t\,\sigma_t^2\,\eta_t$ and $t\,\sigma_t^{2\alpha}\,\eta_t^\alpha\,/k_t^{1-\alpha}$ go to zero as $t$ goes to zero.
		\end{enumerate} 
		
	\end{theorem}
	\begin{proof}
		See Appendix A
	\end{proof}
	
	Let us emphasize that
	the conditions of \textbf{Theorem \ref{theorem:f_Additive}} are quite general. In the market, we observe only a limited number of maturities, and thus,  there is a large set of functions of time that reproduce market data and satisfy the conditions of the theorem. Furthermore, we prove in  \textbf{Theorem \ref{theorem:semplified_f}}, that the conditions of \textbf{Theorem \ref{theorem:f_Additive}} are satisfied by a simple sub-case of ATS with power-law scaling $\eta_t$ and $k_t$ and constant sigma. 
	Finally, we also prove in \textbf{Proposition \ref{theorem:NewAdditive}} 
	that ATS models allow a generic volatility term structure $\sigma_t$.
	
	\bigskip
	In a similar way to the LTS case, it is possible to consider
	a forward price $F_t (T)$ (\ref{eq:FwdBasic}) as the exponential of the ATS process $ \left\{f_t \right\}_{t\geq 0}$ with the characteristic function  (\ref{laplace}).
	The deterministic function of time $\varphi_t $ 
	can be chosen s.t. the process $F_t (T)$ satisfies the martingale property, as shown in the next proposition.
	
	\begin{proposition}
		{\bf Martingale property} \label{th6}\\
		The forward $\left\{F_t (T)\right\}_{t\geq 0}$,
		modeled via an exponential additive process characterized by an  
		ATS process	$\left\{f_t\right\}_{t\geq 0}$
		is a martingale
		if and only if 
		\be
		\varphi_t\; t=-\ln {\cal L}_t\left(\sigma^2_t \eta_t;\;k_t,\;\alpha\right)  \;\;,
		\label{eq:drift}
		\en
				where ${\cal L}_t$ is the Laplace transform in (\ref{eq:lap_transf}).
	\end{proposition}
	
	\begin{proof}
		See Appendix A
	\end{proof}
	
	
	We introduce a sub-case of ATS, determined by self-similar functions of time. 
	In Subsection {\bf \ref{section:scaling}}, 
	we show that this family of processes describes accurately 
	market implied volatility surfaces. 
	Power-law scaling functions of time allow us to rewrite {\bf Theorem \ref{theorem:f_Additive}} conditions as simple inequalities on the scaling parameters.
	
	\begin{theorem}{\bf Power-law scaling ATS\\}\label{theorem:semplified_f}
		There exists an ATS with
		\[
		k_t=\bar{k} \; t^\beta, \qquad \eta_t=\bar{\eta} \; t^\delta,  \qquad\sigma_t=\bar{ \sigma}\;,
		\]
		where  $\alpha \in [0,1)$, $\bar{\sigma}, \bar{k}, \bar{\eta} \in \mathbb{R}^+$, and $\beta, \delta \in \mathbb{R}$
		that satisfy the following conditions:
		\begin{enumerate}
			\item $ \displaystyle 0 \leq \beta\leq \frac{1}{1-\alpha/2}\;\;;$ 
			\item $-\min\left(\beta, \dfrac{1-\beta\left(1-\alpha\right)}{\alpha} \right)<\delta\leq 0\;\;;$
		\end{enumerate}
		where the second condition reduces to $ -\beta< \delta \leq 0$ for $\alpha =0$.
	\end{theorem}
	\begin{proof}
		See Appendix A
	\end{proof}
	It is interesting to observe that the LTS case falls in the subcase described by this theorem. This corresponds to the case with both $k_t$ and $\eta_t$ time independent;
	that is, $\beta$ and $\delta$ equal to zero.

	\bigskip 
	
	The following result allows us to obtain a new additive process from a known one with a deterministic time change. 
	\begin{proposition} {\bf Deterministic time change of additive process} \label{theorem:NewAdditive}\\ 
		Given an additive process  $\left \{X_t \right \}_{t\geq 0}$ and  a real continuous increasing function of time ${r}_t$ s.t. $r_0=0$, 
		then $\left \{X_{r_t} \right \}_{t\geq 0}$ is an additive process.
	\end{proposition}
	\begin{proof}
		See Appendix A
	\end{proof}
	
	Thanks to \textbf{Proposition \ref{theorem:NewAdditive}}, it is possible to extend the ATS power-law scaling sub-case to a case with time-dependent $\sigma_t$. Indeed, if $\sigma_t^2 t$ is non decreasing, we can use it to time-change a power-law scaling ATS without losing the property of independent increments: being able to reproduce the volatility term structure is an important feature from a practitioner perspective.
	
	\smallskip
	In the next Section, we show that the ATS model introduced in this Section describes accurately volatility surfaces observed in the equity derivative market. 
	
	\section{Model calibration and power law scaling}
	\label{section:calibration}
	In this Section, we show that the ATS processes achieve excellent calibration results on the S\&P 500 and EURO STOXX 50 
	volatility surfaces; moreover,
	we show that power-law scaling parameters are observed in market data.  
	
	First, after having described the dataset, 
	we illustrate the model calibration procedure and compare the performance of ATS processes with some benchmarks (LTS processes and  Sato processes in
	\citet{carr2007self}). Then, 
	we outline some statistical evidence that the market-implied volatility surface is compatible with a power-law scaling of ATS parameters.
	
	\subsection{Dataset} 
	\label{subsection:dataset}
	
	We analyze all quoted  S\&P 500 and EURO STOXX 50 option prices observed at 11:00 am New York Time on the $30^{th}$ of May 2013. The dataset is composed of real market quotes (no smoothing or interpolation).
	Let us recall that the options on these two indices are the most liquid options in the equity market at the world level. 
	For both indices, options expire on the third Friday of
	March, June, September, and December in the front year and June and December in the next year. 
	In the EURO STOXX 50 case also December contracts for the following three years are available.
	The dataset includes the risk-free interest rate curves bootstrapped from (USD and EUR) OIS curves.
	Financial data are provided by Bloomberg.  
	The dataset contains all bid/ask prices for both call and put. The strikes are in a regular grid for each available maturity. We exclude options that do not satisfy two simple liquidity thresholds. We discard options whose price is less than 10\% the minimum difference in the grid of market strikes (the so-called \textit{penny options}) and options with bid-ask over bid bigger than $60\%$. The last condition filters out strikes for which either a bid or an ask price is missing.
	\bigskip
	
	
	We use the synthetic forward, as forward price, because this allows a perfect synchronization with option prices and, for several maturities, it identifies the most liquid forward in the market. 
	The synthetic forward price is obtained for every maturity, from very liquid options as the (algebraic) mean of the lowest forward ask and the highest forward bid. 
	
	\smallskip
	We implement a simple iterated algorithm that identifies the synthetic forward price at a given maturity $T$: let us briefly describe it.
	We start selecting the call and put options with strike price nearest to the spot price for the shortest maturity or to the previous maturity forward price for the next maturities. 
	We compute forward bid, ask, and mid prices for that strike price.
	We consider the options with the nearest superior strike. If the forward mid-price computed previously falls within the new bid-ask interval, 
	then the updated forward bid is the highest value among the two forward bids, 
	while the updated ask is the lowest value among the two ask prices. 
	The updated forward mid-price corresponds to the mean of the updated bid and ask. 
	Then, we consider the nearest inferior strike and iterate the same procedure comparing the updated forward mid-price with the new bid/ask prices relative to this new strike.
	This procedure is iterated with the next superior strike  and then with the next    inferior 
	strike, and so on for all the options present at that maturity $T$.
	
	\smallskip
	In Figure \ref{fig:forward curve},
	we show, for a given underlying and a given maturity, the values considered in the forward price construction and the value selected by the procedure. 
	\begin{center}
		\begin{minipage}[t]{1\textwidth}
			\centering
			{\includegraphics[width=.90\textwidth]{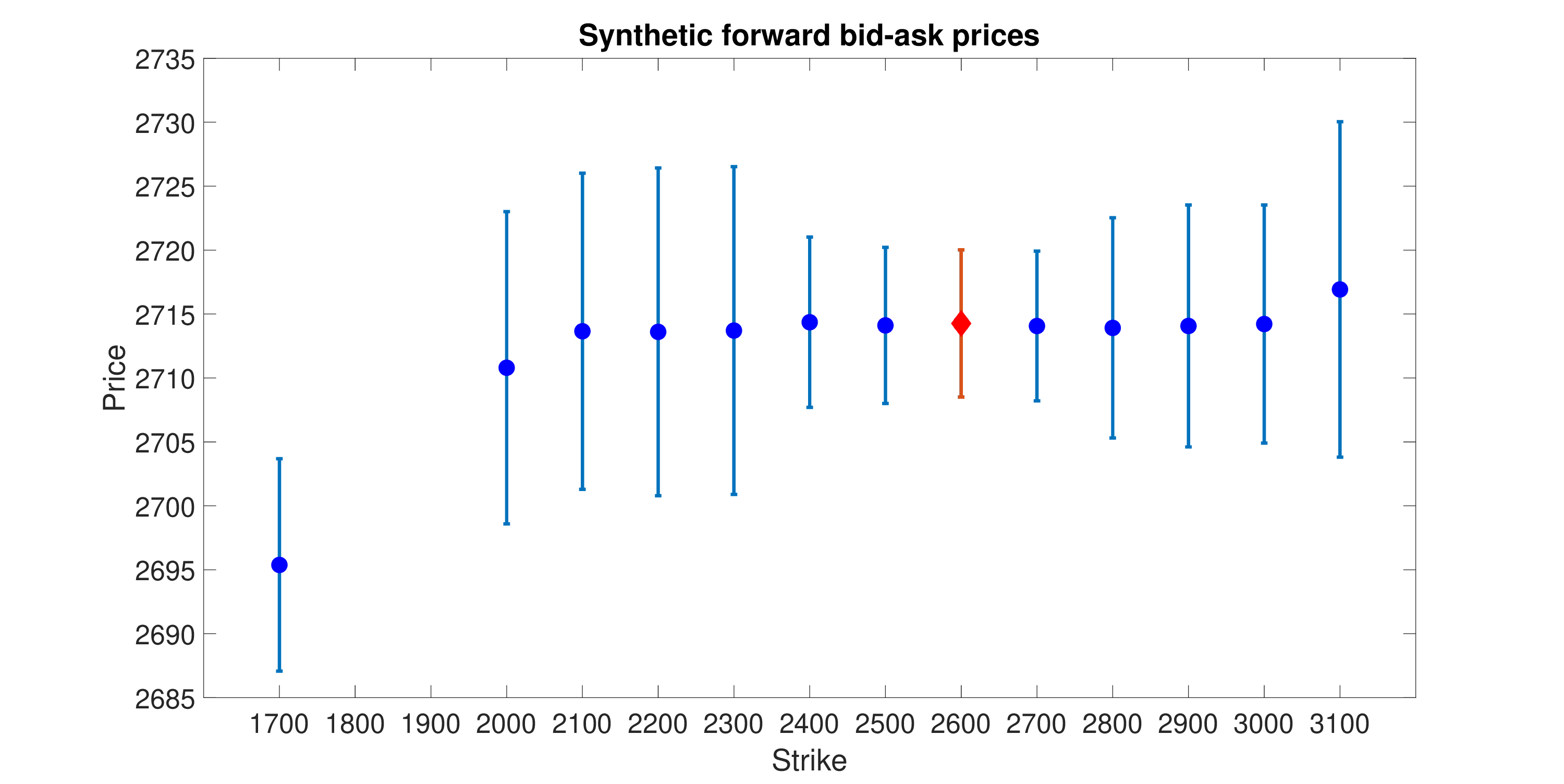}} 
			\captionof{figure}{\small 
		 EURO STOXX 50 	synthetic forward prices on the $30^{th}$ of May 2013 at 11 am NT for the JUN14 maturity:
				bid, ask, and mid forward prices. Only prices not discarded by the two liquidity criteria are shown in the figure. 
				According to the algorithm described in the text also the price related to the  strike  $1700$
				is discarded from the forward price computation.
				We show in red the corresponding forward bid-ask prices and with a diamond, the selected forward price $F_0(T)$ relative to this expiry. 
			}
			\label{fig:forward curve}
		\end{minipage}
	\end{center}
	
	In Figure \ref{fig:forward term struct}, we plot the  bid, ask, and mid synthetic forward prices for the different maturities available
	for the S\&P 500 and the EURO STOXX 50. 
	
	\begin{center}
		\begin{minipage}[t]{0.9\textwidth}
			\includegraphics[width=\textwidth]{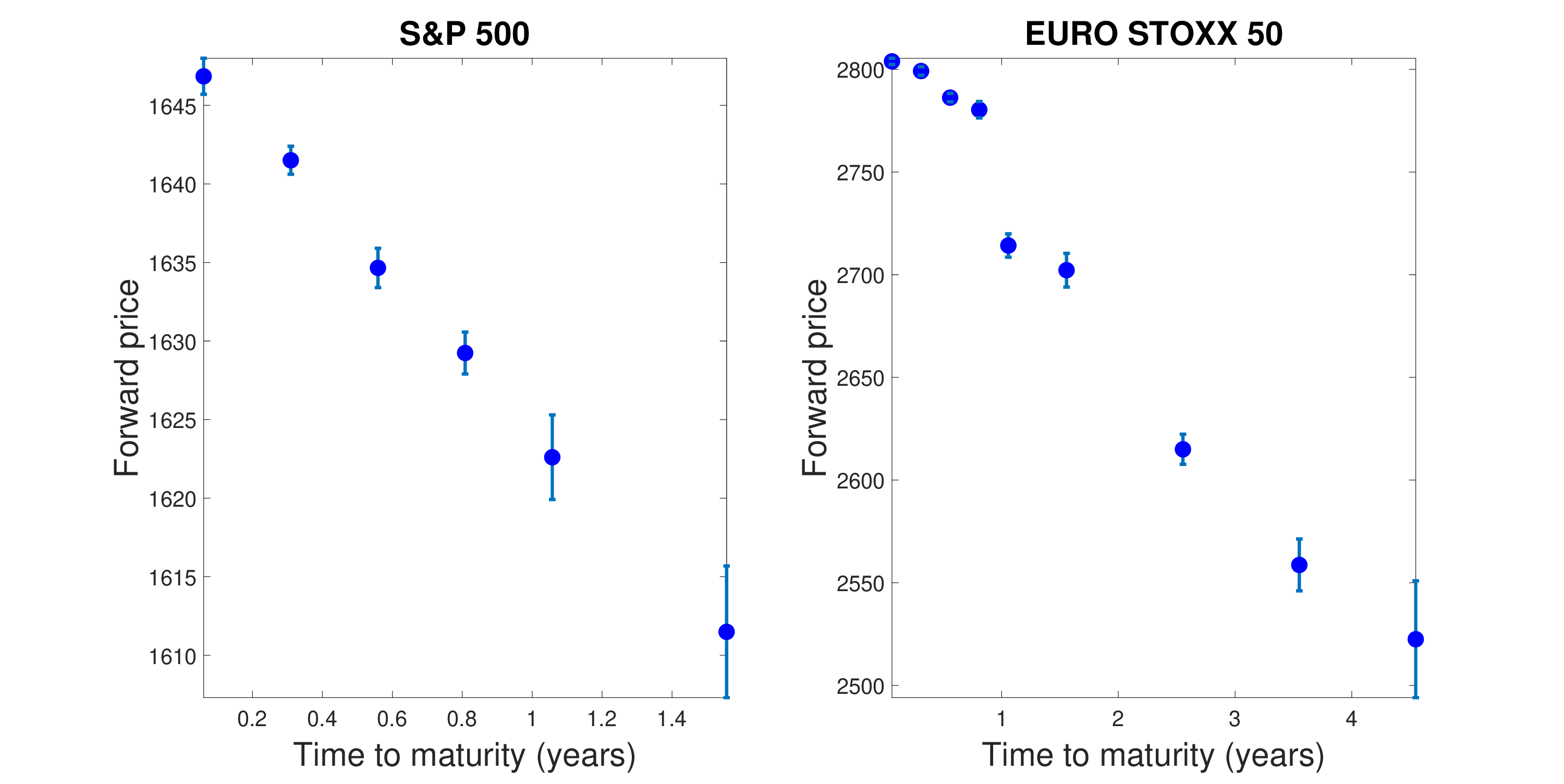} 
			\captionof{figure}{Term structure of the synthetic     forward prices on the $30^{th}$ of May 2013 : 
				we report also the observed bid  and ask prices for every maturity. 
				On the left hand side, we plot  the S\&P 500 index case and, on the right hand side, the EURO STOXX 50 index     case.} 
			\label{fig:forward term struct}
		\end{minipage}
	\end{center}
	
	\subsection{Calibration}
	\label{subsection:model calibration}
	
	
	We calibrate the ATS following the procedure discussed by \citet[][Ch.14, pp.464-465]{Cont}.  We cut the volatility surface into slices, each one containing options with the same maturity, and calibrate each slice separately. 	Hereinafter, we focus on $\alpha = 1/2$ (NIG) and $\alpha =0$ (VG), which are the two (ATS and Sato) generalizations of the two most frequently used LTS processes.
	For every fixed maturity $T$, it is possible to define a new  L\'evy normal tempered stable process such that, at time $T$, 
	its marginal distribution is equal to the marginal distribution of an ATS.
	A different L\'evy NIG and VG is calibrated for every different maturity and the three time-dependent parameters $k_T, \eta_T, \sigma_T$ are obtained. The calibration is performed imposing the conditions of monotonicity of \textbf{Theorem \ref{theorem:f_Additive}}.

	Beneath the ATS processes, we consider the calibration of the standard L\'evy processes and of the (four parameters) Sato processes proposed by \citet{carr2007self}.\footnote{We underline that, in both cases (LTS and Sato), model parameters are obtained through a global calibration of the whole volatility surface.} 
We remind that the latter are additive and self-similar processes \citep[see, e.g.][]{sato1991self}.
Call option prices, with strike $K$ and maturity $T$, are computed using the \citet{lewis2001simple} formula
	\begin{equation}
	C (K,T)   =   {B_T \; F_{0} (T) } \left\{
	{\displaystyle 1 - e^{{x}/2} 
		\int^{\infty}_{-\infty} \frac{d  z}{2 \pi}
		e^{ i  z \, x} \phi^c \left( -z - \frac{i}{2} \right)  \frac{1}{ z^2 + \frac{1}{4}} } 
	\right\}\;\,,
	\label{eq: Lewis}
	\end{equation}
	where $\phi^c ( u) $ is the characteristic function of  $f_T$,
	$x := \ln K/F_0 (T) $ is the {\it moneyness}, and $B_T$ is the discount factor between value date and $T$. 
	
	
	The calibration is performed by minimizing the Euclidean distance between model and market prices. 
	The simplex method is used to calibrate every maturity of the ATS process.
	For L\'evy processes and Sato processes, 
	because standard routines for global minimum algorithms  are not satisfactory, 
	we consider 
	a differential evolution algorithm together with a  multi-start simplex method.
	
	The calibration performance is reported in Table \ref{tab: MSE_APE} in terms of Mean Squared Error (MSE) and Mean Absolute Percentage Error (MAPE).\footnote{Calibrated model parameters are available upon request.}  
	It is possible to observe that Sato processes slightly improve L\'evy performance, as reported in the literature \citep[see, e.g][]{carr2007self},
	while the ATS processes improvement is, on average, above two orders of magnitude.  Although we present the results for VG and NIG, 
	similar results can be obtained for all ATS processes with $\alpha \in [0, 1)$. The worst results are observed in the VG case.

	\begin{center}
		\begin{tabular} {|cc|ccc|ccc|}
			\toprule
			& & \multicolumn{3}{c|}{MSE} &\multicolumn{3}{c|}{MAPE} \\
			\hline
			Index & Model&  L\'evy &Sato & ATS  & L\'evy &Sato  & ATS   \\ 
			\hline 
			S\&P 500&NIG & ${ 4.56}$ & $ 1.92  $ &${\bf 0.02}$ &$3.13\%$ & $ 1.47\%   $ &${\bf 0.23\%}$\\
			S\&P 500 &VG &  $8.49$ & $2.20$ &${\bf 0.24}$ & $ 4.31\%$ & $1.62\%$ &${\bf 0.79 \%}$ \\
			Euro Stoxx 50 &NIG & $22.15$ & $ 9.87$ &${\bf 0.10}$& $1.75\%$ & $ 0.75\%$ &$ {\bf 0.09\%}$\\
			Euro Stoxx 50 &VG &  $55.81$ & $9.22$ &$ {\bf 0.35}$ & $2.85\%$ & $0.73\%$ &${\bf 0.21\%}$ \\
			\bottomrule
		\end{tabular}
		
		\captionof{table}{\small Calibration performance for the S\&P 500 and EURO STOXX $50$
			in terms of MSE and MAPE.
			In the NIG ($\alpha=1/2$) and VG ($\alpha=0$) cases, we consider the standard L\'evy process, the Sato process, and the corresponding ATS process.  
			Sato processes perform better than L\'evy processes but ATS improvement is far more significant: 
			two orders of magnitude for MSE and one order of magnitude for MAPE. 
		}		\label{tab: MSE_APE}
		
	\end{center}
	Figure \ref{figure:MSEcomparison} shows 
	the differences of MSE w.r.t. the different times to maturity for S\&P 500 volatility surface calibrated with a NIG process. 
	Sato and L\'evy LTS have a MSE of the same order of magnitude, while the improvement of ATS is of two orders of magnitude and particularly significant at short-time. 
	The short time improvement in implied volatility calibration is particularly evident,  
	as shown in Figure \ref{figure:vol}.
	\begin{center}
		\begin{minipage}[t]{1\textwidth}
			\centering
			\includegraphics[width=0.9\textwidth]{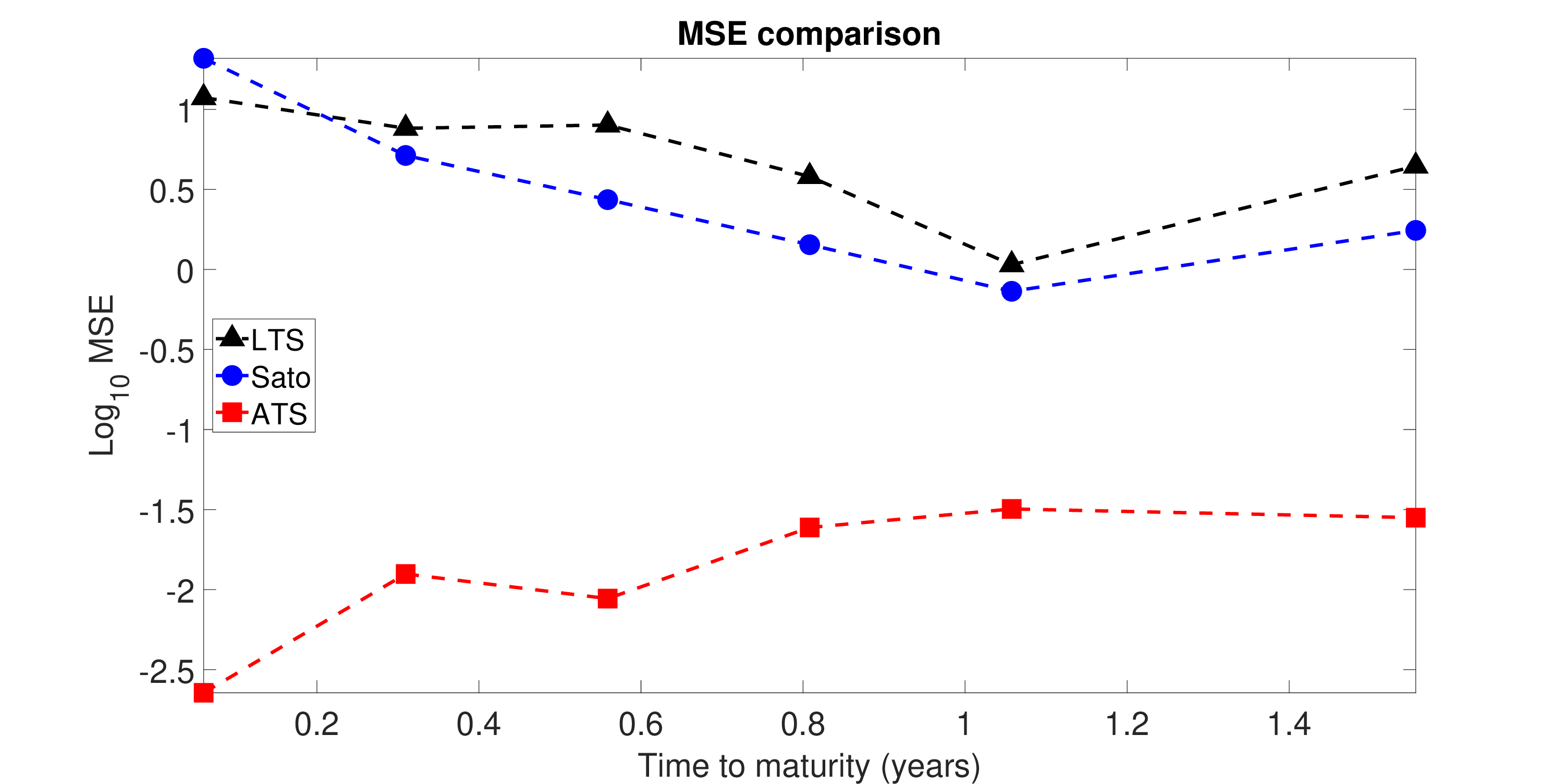} 
			\captionof{figure}{\small MSE w.r.t. the different times to maturity (in years) for S\&P 500 volatility surface calibrated with a NIG process. 
				Sato (circles) and L\'evy (triangles) have a MSE of the same order of magnitude, while the improvement of ATS (squares) is of two orders of magnitude and particularly significant at short-time.} \label{figure:MSEcomparison}
		\end{minipage}
	\end{center}
	
	
	In Figure \ref{figure:vol},
	we plot the market implied volatility and the volatility replicated via ATS, LTS, and Sato processes at the 22 days (on the left) and 9 months and 21 days (on the right) maturities. 
	We observe that the ATS implied volatility is the closest to the market implied volatility in any case and it significantly improves
	both LTS and Sato processes, particularly for small maturities.  Similar results hold for all other ATS. 
	\begin{center}
		\begin{minipage}[t]{1\textwidth}
			\centering
			\includegraphics[width=0.9\textwidth]{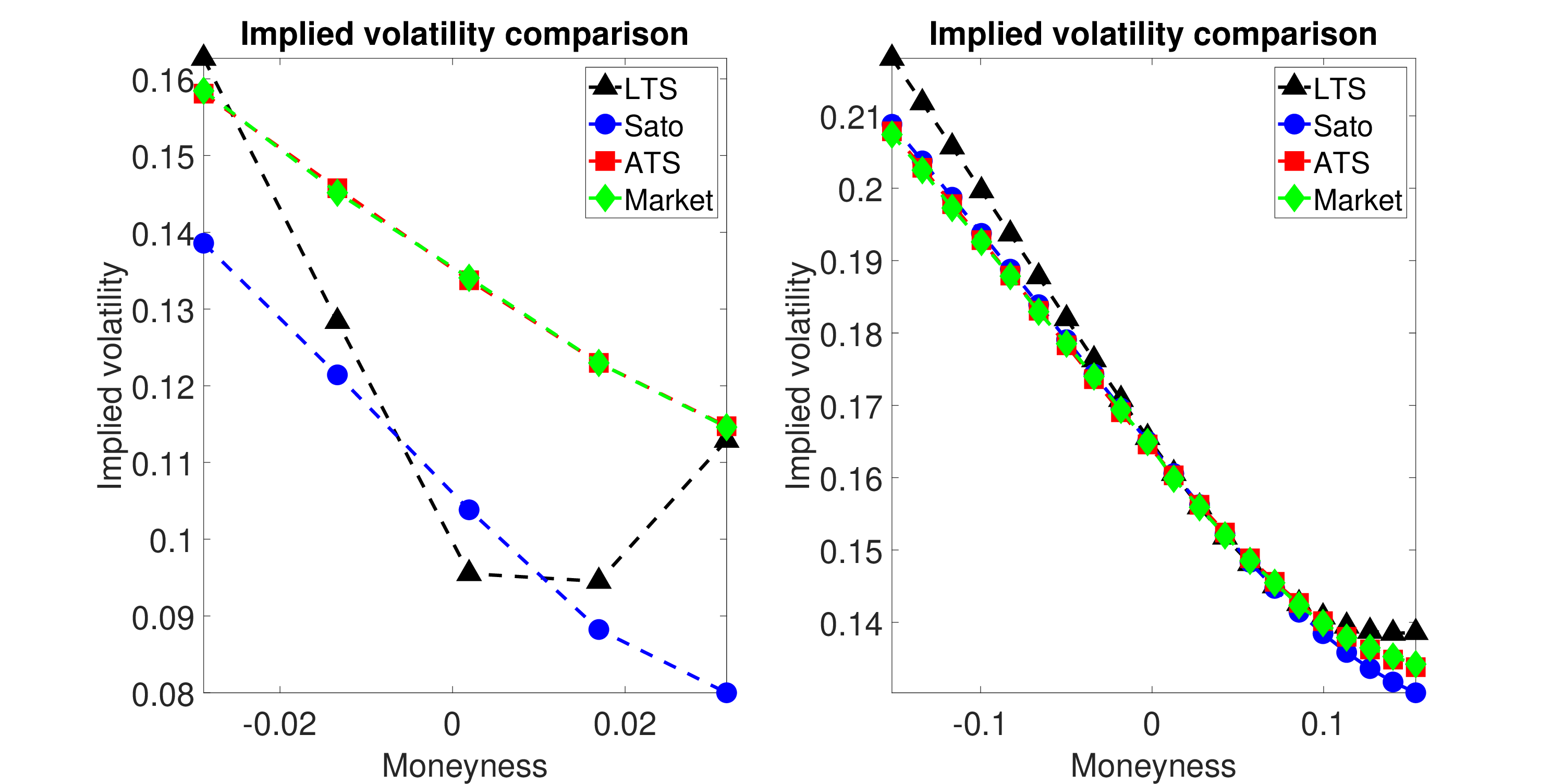} 
			\captionof{figure}{\small 
				Implied volatility smile for S\&P 500 for a given time to maturity: 22 days (on the left) and 9 months and 21 days (on the right). 
				The NIG ATS process, Sato process, and LTS process implied volatilities are plotted together with the market-implied volatility. 
				ATS reproduces the smile significantly better than the alternatives, the improvement is particularly evident for small maturities.}\label{figure:vol}
		\end{minipage}
	\end{center}
	In Figure \ref{figure:skew}, we plot the market and the ATS implied volatility \textit{skew}  for EURO STOXX 50 w.r.t. the times to maturity. We observe that the calibrated ATS replicates accurately the market implied volatility \textit{skew}.
	\begin{center}
		\begin{minipage}[t]{1\textwidth}
			\centering
			\includegraphics[width=0.9\textwidth]{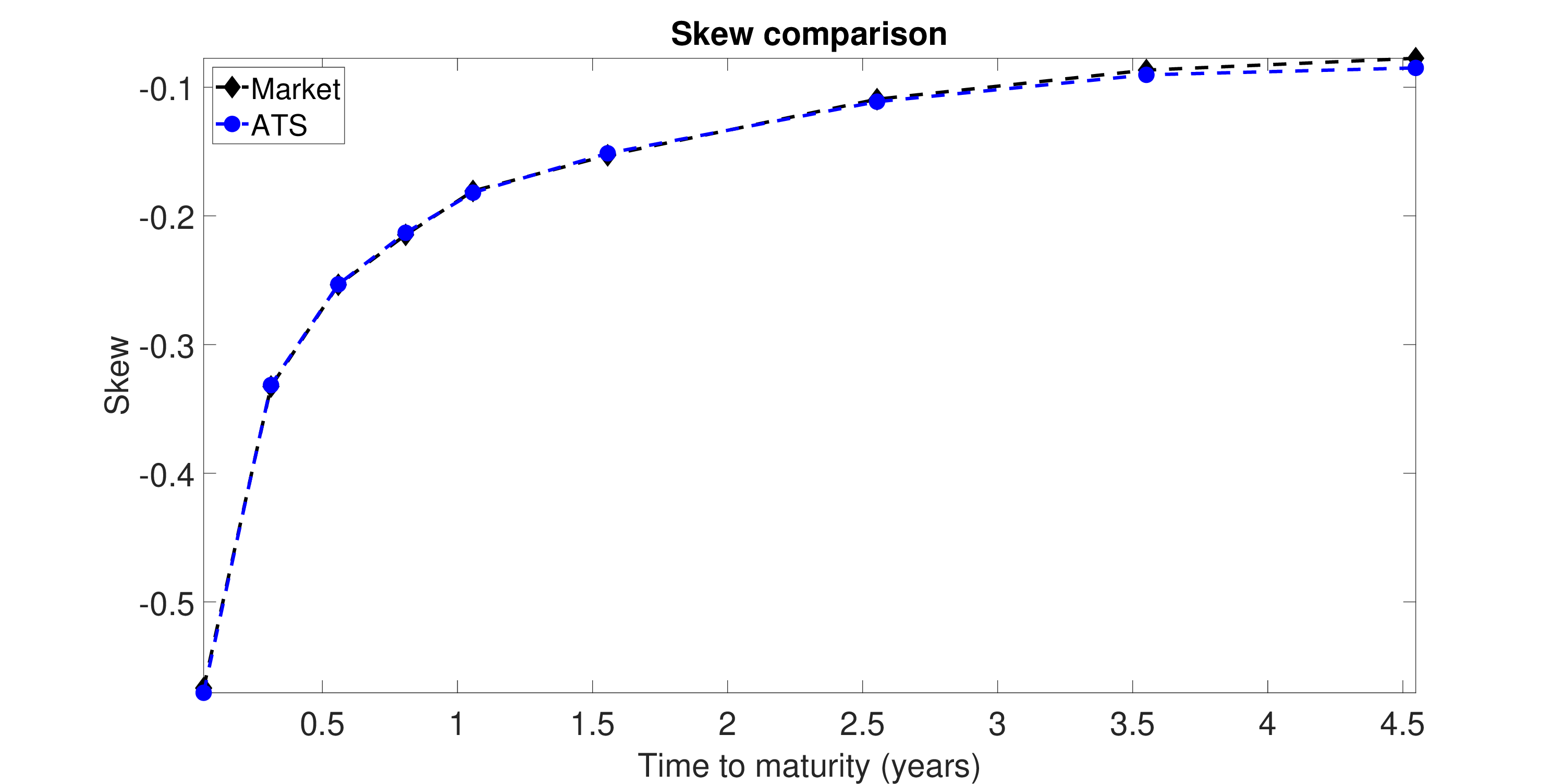} 
			\captionof{figure}{\small 
				The market and the NIG ATS implied volatility \textit{skew} for EURO STOXX 50 w.r.t. the times to maturity. ATS replicates the market implied volatility \textit{skew} behavior.}\label{figure:skew}
		\end{minipage}
	\end{center}
	The calibration results of ATS are startling. In particular, we have a model that reproduces “exactly” the volatility term structure observed in the market.\\
	
	It is useful to stop and comment. The ATS model allows us to calibrate slice-by-slice the surface; we have only to impose the monotonicity conditions of \textbf{Theorem \ref{theorem:f_Additive}}. With the slice-by-slice approach, we use 3 parameters for every expiry (e.g. 18 parameters in the S\&P case). We have observed that ATS outperforms Levy and Sato processes, the benchmark pure-jump processes in the literature. It could seem unfair to compare the calibration results of ATS with a Levy (3 parameters) and a Sato (4 parameters).\\
	
	In the next Subsection, we show that this family of additive processes combines parsimony with the desired property of a perfect fit of the volatility term structure: we show that, once the term structure has been taken into account, only 2 free parameters allow a detailed calibration of the whole volatility surface.
	\subsection{Scaling properties}
	\label{section:scaling}
	In this Subsection, we show that power-law scaling parameters are observed in market data.
	This stylized fact is extremely relevant:
	we observe statistical evidence that the market-implied volatility surface is compatible with a power-law scaling ATS of \textbf{Theorem \ref{theorem:semplified_f}}. 
	
	\bigskip 
	We introduce a new ATS process, w.r.t. the time $\theta:=T \sigma^2_T$.\\
	We define $\hat{k}_\theta:=k_T \sigma_T^2$ and $\hat{\eta}_\theta:=\eta_T$. We call  $\hat{ f}_{\theta}$ the ATS with parameters $ \hat{k}_\theta$, $\hat{\eta}_\theta$ and $\sigma_\theta:=1$. Notice that for every maturity $\hat{ f}_{\theta}$ has the same characteristic function of the calibrated ATS $f_T$.
	We analyze the re-scaled parameters, in both S\&P 500 and EURO STOXX 50 cases.
	We observe a self-similar behavior of ${\hat k}_\theta$ and ${\hat \eta}_\theta$; that is,
	\begin{equation}
	\begin{cases}
	\hat{k}_\theta &=\;  \bar{k}\theta^\beta \\
	\hat{\eta}_\theta &=\;\bar{\eta}\theta^\delta 
	\end{cases}\; ,
	\label{eq:scaling}
	\end{equation}
	where $\bar{k}$, $\bar{\eta}$ are positive constants and $\beta$, $\delta$ are real constant parameters. 
	To investigate this behavior and to infer the value of the scaling parameters we consider equations (\ref{eq:scaling}) in the log-log scale. 
	
	In Figures \ref{figure:SPPNIG} and \ref{figure:EUVG} we plot the weighted regression lines and the observed time-dependent parameters $\ln {\hat k}_\theta$ and $\ln {\hat \eta}_\theta$ 
	with their confidence intervals 
	for S\&P $500$ and EURO STOXX $50$.
	The confidence intervals are two times the standard deviations of $\ln {\hat k}_\theta$,  of
	$\ln{\hat \eta}_\theta$ and of $\ln \theta$. 
	In Appendix B, we discuss the estimation of the standard deviations via a confidence interval propagation technique
	and the selection of the weights.
	\begin{center}
		\begin{minipage}[t]{1\textwidth}
			\centering
			\includegraphics[width=0.9\textwidth]{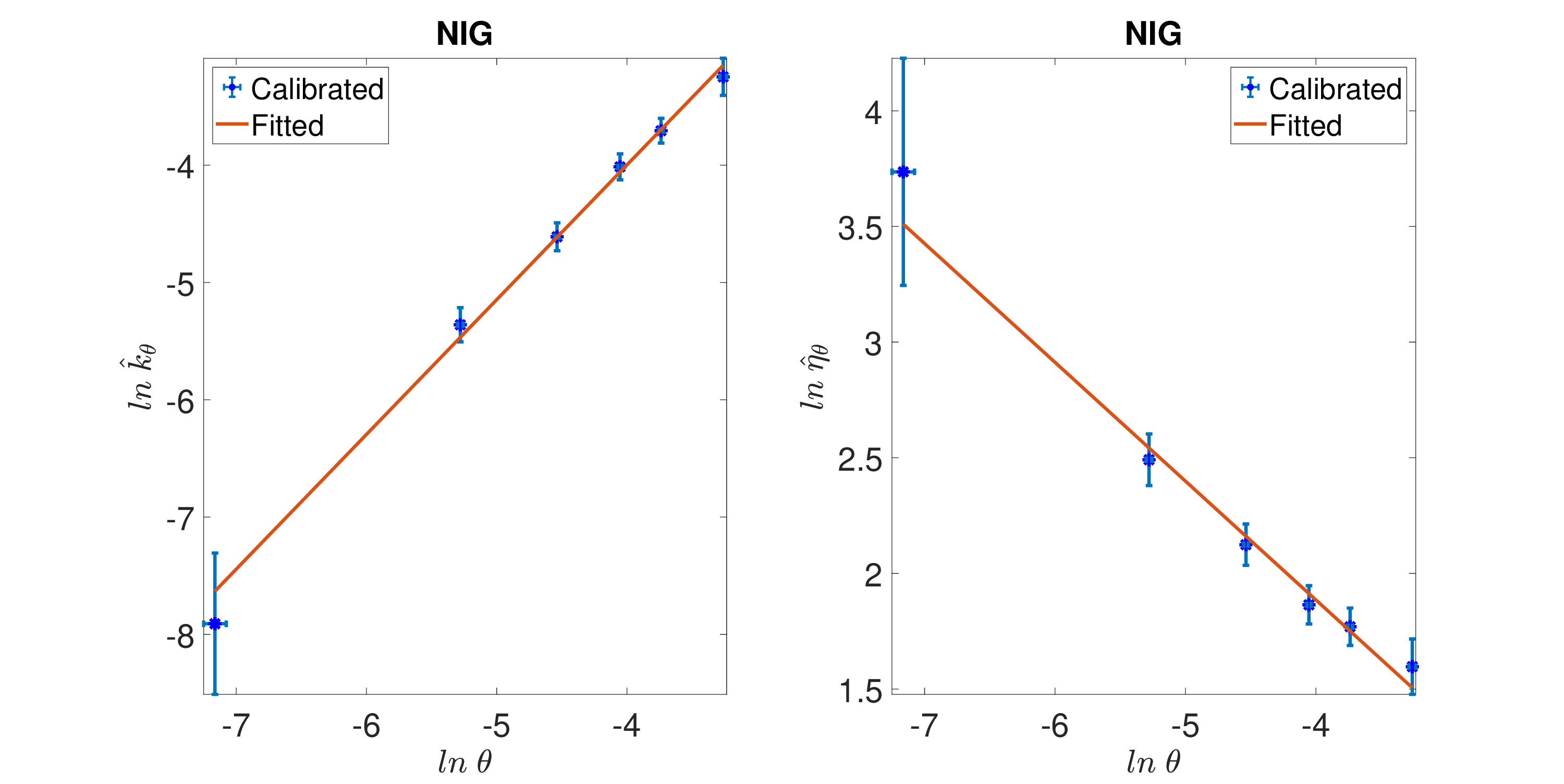} 
			\captionof{figure}{\small Weighted regression line and the observed time-dependent parameters $\ln {\hat k}_\theta$ and $\ln {\hat \eta}_\theta$ w.r.t. $\ln\theta$ for the NIG calibrated model for S\&P 500. We plot a confidence interval equal to two times the corresponding standard deviation. Notice that  confidence intervals on $\ln {\hat k}_\theta$  and $\ln {\hat \eta}_\theta$  are one order of magnitude wider than confidence intervals on $\ln \theta$. The scalings of $\hat{k}_\theta$ and $\hat{\eta}_\theta$ in (\ref{eq:scaling}) are statistically consistent with $\beta=1$ and $\delta=-1/2$. The values of $\theta$ correspond to times to maturity that goes from 22 days  to two years.}\label{figure:SPPNIG}
		\end{minipage}\\ \bigskip
		\begin{minipage}[t]{1\textwidth}
			\centering
			\includegraphics[width=0.9\textwidth]{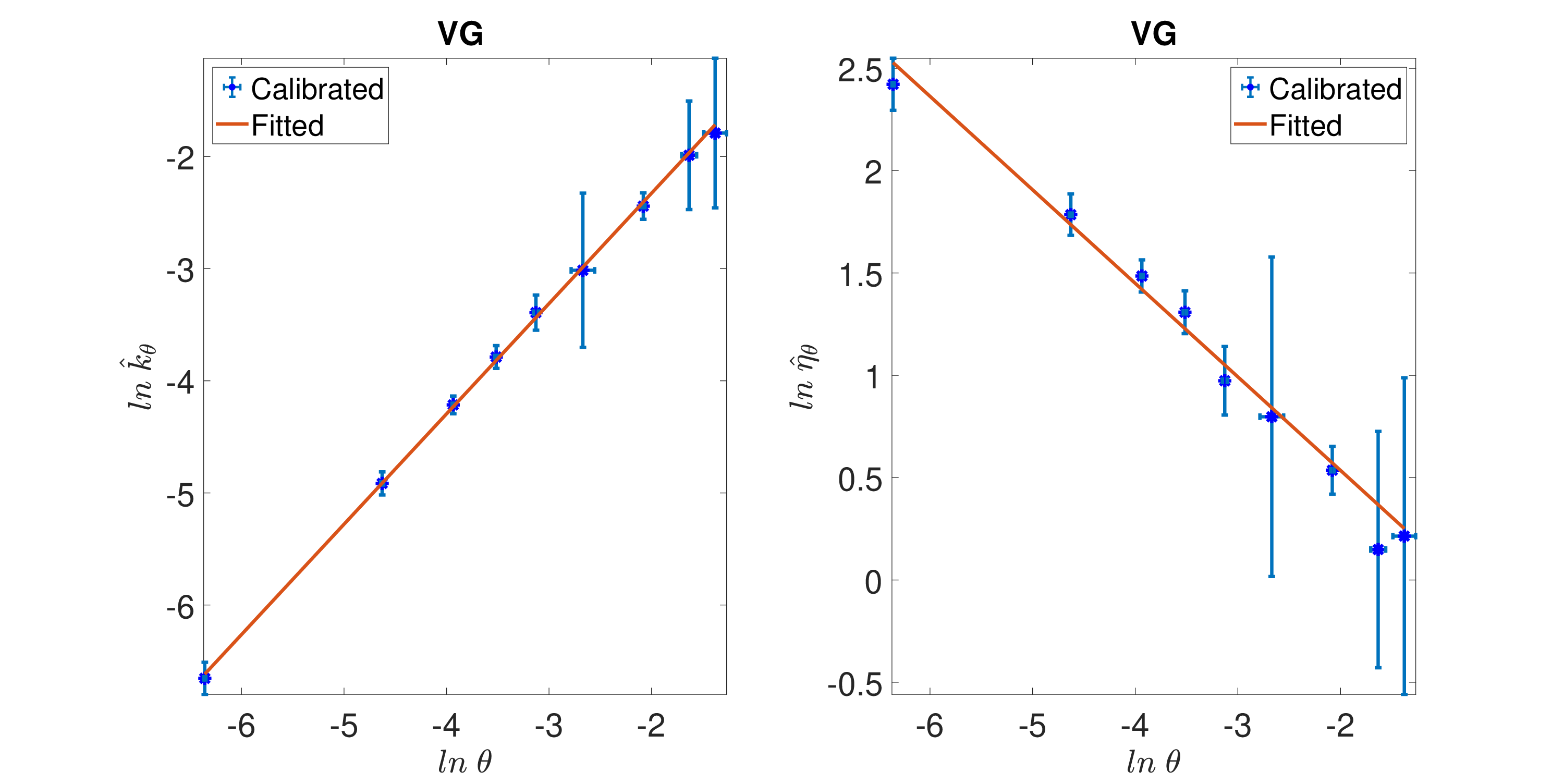}
			\captionof{figure}{\small	Weighted regression line and the observed time-dependent parameters $\ln {\hat k}_\theta$ and $\ln {\hat \eta}_\theta$ w.r.t. $\ln \theta$ for the VG model calibrated  on EURO STOXX 50. We plot a confidence interval equal to two times the corresponding standard deviation.  Notice that  confidence intervals on $\ln {\hat k}_\theta$  and $\ln {\hat \eta}_\theta$  are one order of magnitude wider than confidence intervals on $\ln \theta$. The scalings of $\hat{k}_\theta$ and $\hat{\eta}_\theta$ in (\ref{eq:scaling}) are statistically consistent with $\beta=1$ and $\delta=-1/2$.  The values of $\theta$ correspond to times to maturity that goes from 22 days to five years.}\label{figure:EUVG}
		\end{minipage}\\ \bigskip	
	\end{center}

	We have observed what seems to be a stylized fact of this model class: both ${\hat \eta}_\theta$ and ${\hat k}_\theta$ scale with power-law. The same scaling laws are observed both for short time-horizon (days) and long time-horizon (few years) options.
	The fitted regression lines provide us with an estimation of $\beta$ and $\delta$. 
	Moreover, let us emphasize that the scaling parameters appear qualitatively compatible to
	$\beta=1$ and $\delta=-\frac{1}{2}$ in all observed cases.

	We can test whether there is statistical evidence that our hypotheses are consistent with market data. The estimated scaling parameters together with the p-value of statistical tests are reported in Table \ref{tab:stat tests}. 
	
	\begin{center}
		\begin{tabular} {cccccc}
			\toprule
			Surface & Model&Parameter & Parameter's Value&p-value  \\ 
			\toprule
			S\&P 500&NIG & $\beta$ & $      1.10$ &$   0.228$ \\
			S\&P 500 &NIG &  $\delta$ & $ -0.47$ &$0.705$  \\
			S\&P 500&VG& $\beta$ & $1.01$ &$ 0.758$  \\
			S\&P 500&VG& $\delta$ & $ -0.43$ &$ 0.057$  \\
			EURO STOXX 50&NIG & $\beta$ & $1.02$ & $0.816$ \\
			EURO STOXX 50&NIG &  $\delta$ & $-0.44$ &$  0.472$  \\
			EURO STOXX 50 &VG&$\beta$ & $0.99   $ &$0.690$ \\
			EURO STOXX 50 &VG& $\delta$ & $-0.45$ &$0.195$ \\			
			\bottomrule
		\end{tabular}
		\captionof{table}{\small Scaling parameters calibrated from S\&P 500 and EURO STOXX 50 volatility surfaces for NIG ($\alpha=1/2$) and VG ($\alpha=0$). Parameter estimates are provided together with the p-values of the statistical tests that verify whether it is possible to accept the null hypotheses $\beta=1$ and $\delta=-\frac{1}{2}$.}
		\label{tab:stat tests}
	\end{center}
	
	In all cases, we accept the null hypotheses ($\beta=1$ and $\delta=-\frac{1}{2}$) with a 5\% confidence level. Notice that all p-values, except the S\&P 500 VG $\delta$, are above 19\%.

	In Table \ref{tab:stattests_int} we report an estimation of the parameter $\bar{k}$ and $\bar{\eta}$.
	\begin{center}
		\begin{tabular} {cccccc}
			\toprule
			Surface & Model&Parameter & Parameter's Value&p-value  \\ 
			\toprule
			S\&P 500&NIG & $\bar{k}$ & $      1.50$ &$   0.022$ \\
			S\&P 500 &NIG &  $\bar{\eta}$ & $ 0.98$ &$0.015$  \\
			S\&P 500&VG& $\bar{k}$ & $1.01$ &$ 0.001$  \\
			S\&P 500&VG& $\bar{\eta}$ & $ 0.91$ &$ 0.000$  \\
			EURO STOXX 50&NIG & $\bar{k}$ & $0.68$ & $0.023$ \\
			EURO STOXX 50&NIG &  $\bar{\eta}$ & $1.21$ &$ 0.021$  \\
			EURO STOXX 50 &VG&$\bar{k}$ & $0.98   $ &$0.000$ \\
			EURO STOXX 50 &VG& $\bar{\eta}$ & $0.72$ &$0.000$ \\			
			\bottomrule
		\end{tabular}
		\captionof{table}{\small  $\bar{k}$ and $\bar{\eta}$ calibrated from S\&P 500 and EURO STOXX 50 volatility surfaces. Parameter estimates are provided together with the p-values of the statistical tests that verify whether it is possible to accept the null hypothesis $\bar{k}=0$ and $\bar{\eta}=0$.}
		\label{tab:stattests_int}
	\end{center}
	We have statistical evidence that in all cases $\bar{k}$ and $\bar{\eta}$ are positive (we reject the null hypotheses of $\bar{k}=0$ and $\bar{\eta}=0$ with a 5\% confidence level). From these results and from Figure \ref{figure:skew} it is possible to infer a connection between a positive $\bar{\eta}$ and the observed negative \textit{skew}.\\
	
	It is interesting to observe that these estimated parameters satisfy the inequalities of {\bf Theorem \ref{theorem:semplified_f}} for the existence of a power-law scaling ATS $\hat{ f}_{\theta}$.
	Moreover, the re-scaled process is additive w.r.t. the ``real" time $T$.
	This fact is a consequence of
	the properties of the volatility term structure $\sigma_T$ (it is always observed on real data that $\sigma^2_T T$ is non-decreasing) 
	and of {\bf Proposition \ref{theorem:NewAdditive}}. 
	This proposition states that 
	if $\left\{\hat{f}_\theta\right\}_{\theta \geq 0}$ is an additive process then $\left\{\hat{f}_{T\sigma_T^2}\right\}_{T\geq 0}$ is an additive process w.r.t. $T$.\footnote{We have also considered a global calibration of the implied volatility surfaces considering the power-law scaling parameters in (\ref{eq:scaling}) with $\beta=1$ and $\delta = -0.5$. The results are of the same order of magnitude of Table \ref{tab: MSE_APE}.}

	\section{Model selection and robustness tests}
	\label{section:AdditionalResult}
	
	In this Section, we show two additional results for the proposed process class.
	First, we compare the ATS with two other additive processes already present in the financial literature and propose some statistical tests able to select the most adequate modeling description of the implied volatility surface.
	Then, we show that the results, described in detail in the previous Section, appear robust over time.
	
	\subsection{Model selection via statistical tests}
	
	In this Subsection, we compare ATS with two classes of additive  processes  already present in the financial literature, 
	the Sato processes \citep[see, e.g.][]{carr2007self} and the  additive processes constructed via additive subordination 
	\citep[see, e.g.][]{li2016Additive}. The comparison is among processes that have the marginal distribution of normal tempered stable type: i.e. with the Sato processes NIGSSD and VGSSD and with the sub-class of ATS constructed through additive subordination. An ATS process can be obtained as a Brownian motion subordinated with an additive subordinator,  as in \citet{li2016Additive}, if and only if $\eta_T$ is constant.\footnote{Proof available upon request.}
	We discuss two features: one related to the $\eta_T$ parameter and another to the skewness and to the excess kurtosis of the calibrated exponential forward.

	\bigskip
	
	A first test is built to verify the adequacy of Sato processes. 
	Given the \textit{index of stability} for the model (e.g. chosen $\alpha$ in the Normal Tempered Stable model), it is possible to compute skewness and kurtosis \citep[see, e.g.][p.129]{Cont}. For example the ATS NIG skewness is \begin{align*}
	\frac{ \mathbb{E}\left[\left(f_T-\mathbb{E}\left[f_T\right]\right)^3\right]}{\left( Var(f_T)\right)^\frac{3}{2}}=-\frac{3{\sigma}_T^4\left({\eta_T} +\frac{1}{2}\right)k_T T +2{\sigma_T}^6\left({\eta}_T +\frac{1}{2}\right)^3k_T^2T}{\left({\sigma_T}^2T+{k_T}T{\sigma_T}^4\left(\eta_T+\frac{1}{2}\right)^2\right)^\frac{3}{2}}\;\;.
	\end{align*}
	A Sato process has skewness and kurtosis constant over time, as it can be deduced by definition \citep[see, e.g.][]{carr2007self}. 
	
	We analyze the term structure of these higher-order moments observed in our dataset adopting the same procedure of  \citet{konikov2002option}. 
	For both indices, we observe a linear behavior of skewness and kurtosis w.r.t. the squared root of the maturity as shown in Figure \ref{figure:NIGSP} in the NIG case. 
	In the Figure, we have plotted also the confidence interval chosen equal to two times the standard deviation, respectively, of the skewness and the kurtosis (cf. Appendix B for the methodology adopted to obtain these standard deviations).
	
	The statistical test is simple. We perform a linear regression statistical analysis of the higher moments behavior w.r.t. the square root of the time to maturity $\sqrt{T}$: 
	we reject the null hypothesis of no slope in all of the cases that we analyze (both indices and both tempered stable models; that is, NIG and VG)  with p-values of the order of $10^{-16}$. Similar results hold in all ATS cases.

	\begin{center}
		\begin{minipage}[t]{1\textwidth} \centering
			\includegraphics[width=0.9\textwidth]{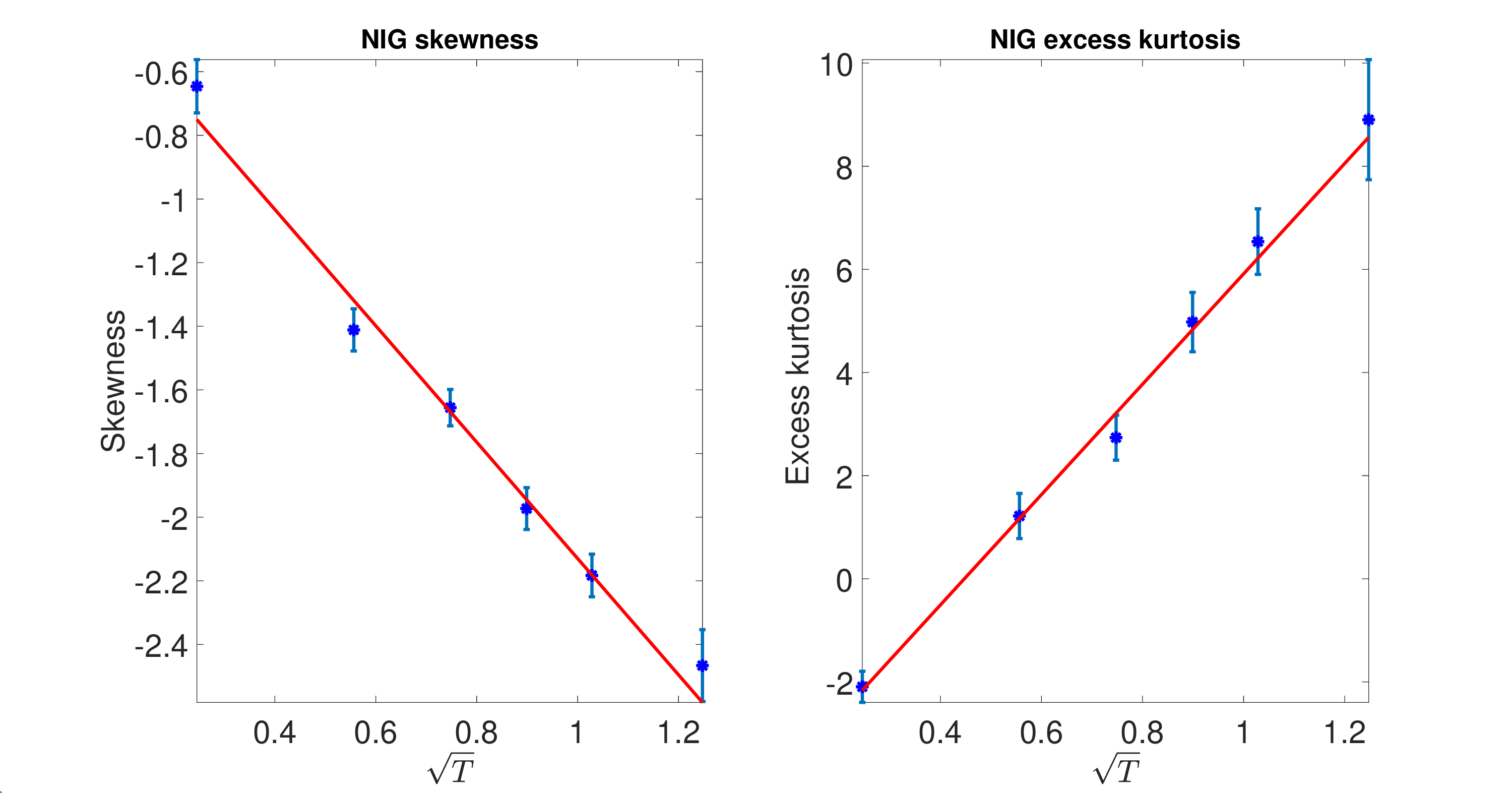}
			\captionof{figure}{\small
				Observed time-dependent  skewness (kurtosis) w.r.t.  $\sqrt{T}$ for the NIG calibrated model on S\&P 500 volatility surfaces. 
				We plot a confidence interval equal to two times the standard deviation.  The behavior is not consistent with a Sato process.}\label{figure:NIGSP}
		\end{minipage}
		\bigskip
		
	\end{center}
	
	The other statistical test aims to verify the 
	adequacy of additive processes obtained through additive subordination \citep{li2016Additive} in volatility surface calibration. 
	As already mentioned the ATS process, when ${\hat \eta}_\theta$ is equal to a constant $\bar{\eta}$, falls within this class.
	
	In Figure  \ref{figure:SPPNIG} and \ref{figure:EUVG} we have already shown the time scaling ${\hat \eta}_\theta$.
	We can statistically test the null hypothesis of constant  ${\hat \eta}_\theta$.
	For both volatility surfaces and for both tested tempered stable models (NIG and VG) we reject the null hypothesis of a constant ${\hat \eta}_\theta$ with p-values below $10^{-7}$.
	As already observed, ATS processes are characterized by a power-law scaling in ${\hat \eta}_\theta$.

	\bigskip
	
	\subsection{Robustness tests}
	In this Subsection, we perform a robustness analysis of the results in Section {\bf \ref{section:calibration}}. We repeat the analysis on four other days, both on the S\&P 500 and EURO STOXX 50 volatility surfaces. We show that the excellent calibration features of the ATS and the power-law scaling properties, observed on the $30^{th}$ of May 2013, arise also in these other dates.\\

	In these robustness tests, we use  bid and ask close prices for the S\&P 500 and EURO STOXX 50 options on the $29^{th}$ of November 2012 (6 months before the date of the analysis of Section {\bf \ref{section:calibration}}, the $30^{th}$ of  May 2013), the $27^{th}$ of February 2013 (3 months before), the $30^{th}$ of August 2013 (3 months after), and the $29^{th}$ November 2013 
	(6 months after).\footnote{These are the penultimate business days of November 2012, February 2013, August 2013, and November 2013.} 
	The dataset includes the bootstrapped risk-free rate curve. The data is provided by Reuters Datastream (option data) and Reuters Eikons (rate data). Let us observe that close prices are, in general, less accurate than open market prices (the ones used for the analysis in Section {\bf \ref{section:calibration}}).\\
	
	In Table \ref{tab:calibration_comparison}, we report calibration performances for the S\&P 500 and EURO STOXX $50$
	in terms of MSE and MAPE on the four dates considered.
	In the NIG and VG cases, we consider the standard L\'evy process, the Sato process, and the corresponding ATS process.  
	In all considered cases, Sato processes perform better than L\'evy processes but ATS improvement is far more significant: on average,
	two orders of magnitude for MSE and one order of magnitude for MAPE. These results appear coherent with the ones reported in Table \ref{tab: MSE_APE}.\\

	\begin{adjustbox}{width=\columnwidth,center}

		\begin{tabular} {|cc|ccc|ccc|ccc|ccc|}
			\toprule

			& & \multicolumn{3}{c|}{MSE} &\multicolumn{3}{c|}{MAPE}& \multicolumn{3}{c|}{MSE} &\multicolumn{3}{c|}{MAPE} \\
			\hline
			Index & Model&  L\'evy &Sato & ATS  & L\'evy &Sato  & ATS & L\'evy &Sato & ATS  & L\'evy &Sato  & ATS   \\ 
			\hline 
			& &  \multicolumn{6}{c|}{$29^{th}$ of November 2012 (-6 months)}&  \multicolumn{6}{c|}{$27^{th}$ of February 2013 (-3 months)}\\ 
			\hline
			S\&P 500&NIG & ${ 4.78}$ & $ 1.15  $ &${\bf 0.38}$ & $ 2.93\%   $&$1.36\%$ &${\bf 0.60\%}$ & ${10.77}$ & $ 4.31  $ &${\bf 0.52}$ &$3.65\%$ & $ 3.30\%   $ &${\bf 0.66\%}$\\
			S\&P 500 &VG &  $11.04$ & $1.00$ &${\bf 0.38}$ & $4.46\%$ & $1.36\%$ &${\bf 0.71 \%}$ & $18.28$ & $3.74$ &${\bf 0.48}$ & $4.81\%$ & $2.16\%$ &${\bf 0.69 \%}$\\
			Euro Stoxx 50 &NIG & $20.64$ & $19.73$ &${\bf 0.26}$& $2.39\%$ & $ 2.29\%$ &$ {\bf 0.18\%}$&$54.54$ & $19.99$ &${\bf 0.15}$& $3.79\%$ & $ 2.47\%$ &$ {\bf 0.15\%}$\\
			Euro Stoxx 50 &VG &  $34.51$ & $20.65$ &$ {\bf 0.41}$ & $3.05\%$ & $2.38\%$ &${\bf 0.31\%}$&  $90.81$ & $19.25$ &$ {\bf 0.38}$ & $4.91\%$ & $2.47\%$ &${\bf 0.24\%}$ \\
			\hline 
			& &  \multicolumn{6}{c|}{ $30^{th}$ of August 2013 (+3 months)}&  \multicolumn{6}{c|}{$29^{th}$ November 2013 (+6 months)}\\ 
			\hline
			S\&P 500&NIG & ${ 8.27}$ & $ 1.08  $ &${\bf 0.18}$ &$3.29\%$ & $ 1.21\%   $ &${\bf 0.12\%}$ & ${ 10.23}$ & $ 1.42  $ &${\bf 0.01}$ &$3.50\%$ & $ 1.32\%   $ &${\bf 0.09\%}$\\
			S\&P 500 &VG &  $18.37$ & $0.98$ &${\bf 0.37}$ & $ 4.95\%$ & $1.16\%$ &${\bf 0.20 \%}$ & $16.80$ & $1.36$ &${\bf 0.09}$ & $ 4.52\%$ & $1.30\%$ &${\bf 0.35 \%}$\\
			Euro Stoxx 50 &NIG & $40.98$ & $ 5.03$ &${\bf 1.53}$& $2.68\%$ & $ 0.93\%$ &$ {\bf 0.44\%}$&$24.22$ & $ 12.03$ &${\bf 0.50}$& $2.38\%$ & $1.73\%$ &$ {\bf 0.27\%}$\\
			Euro Stoxx 50 &VG &  $59.23$ & $4.81$ &$ {\bf 0.64}$ & $3.26\%$ & $0.94\%$ &${\bf 0.32\%}$&  $57.25$ & $12.66$ &$ {\bf 0.91}$ & $3.75\%$ & $1.77\%$ &${\bf 0.45\%}$ \\
			\bottomrule		
		\end{tabular}

	\end{adjustbox}

	\captionof{table}{\small Calibration performance for the S\&P 500 and EURO STOXX $50$
		in terms of MSE and MAPE on the $29^{th}$ of November 2012  (6 months before the date of the analysis), the $27^{th}$ of February 2013 (3 months before), the $30^{th}$ of August 2013 (3 months after), and the $29^{th}$ November 2013 (6 months after).
		In the NIG and VG cases, we consider the standard L\'evy process, the Sato process, and the corresponding ATS process, as in Table \ref{tab: MSE_APE}.  
		In all considered cases, Sato processes perform better than L\'evy processes but ATS improvement is far more significant: 
		two orders of magnitude for MSE and one order of magnitude for MAPE.  \label{tab:calibration_comparison}
	}	\bigskip

	In Figure \ref{figure:MSEcomparison_dates}, we plot the MSE w.r.t. the different times to maturity (in years) for S\&P 500 volatility surface calibrated with a NIG process on the four considered dates. 
	We observe that Sato (circles) and L\'evy (triangles) have a MSE of the same order of magnitude, while the improvement of ATS (squares) is, on average, of two orders of magnitude and particularly significant at short-time.
	\begin{center}

		\begin{minipage}[t]{1\textwidth}
			\centering
			\centerline{\includegraphics[width=1.15\textwidth]{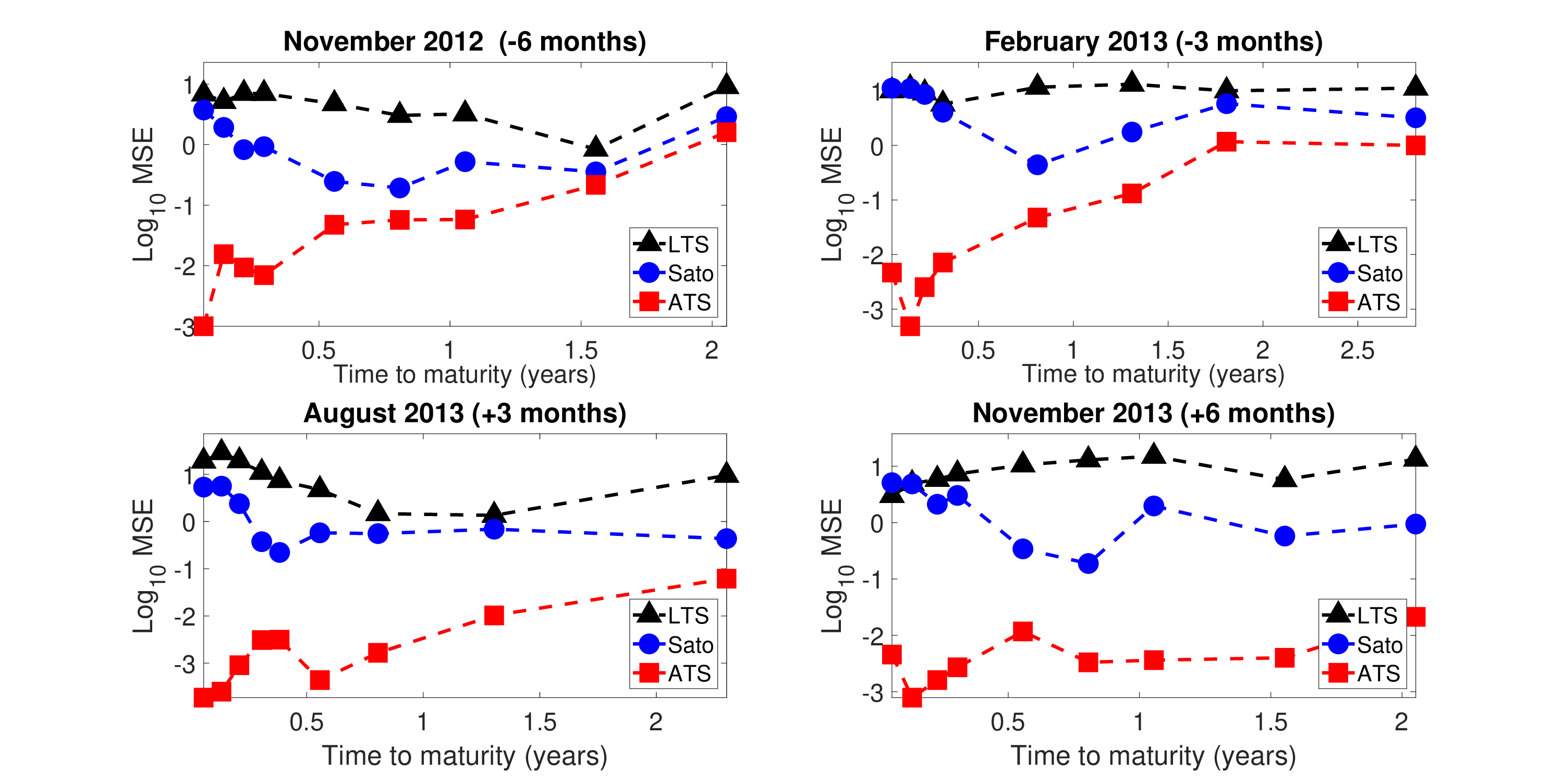} }
			\captionof{figure}{\small MSE w.r.t. the different times to maturity (in years) for S\&P 500 volatility surface calibrated with a NIG process on the $29^{th}$ of November 2012  (6 months before), the $27^{th}$ of February 2013 (3 months before), the $30^{th}$ of August 2013 (3 months after), and the $29^{th}$ November 2013 (6 months after). 
				Sato (circles) and L\'evy (triangles) have a MSE of the same order of magnitude, while the improvement of ATS (squares) is, on average, of two orders of magnitude and particularly significant at short-time.} \label{figure:MSEcomparison_dates}
		\end{minipage}
	\end{center}
	In Figures \ref{figure:volMarch} and \ref{figure:volJuly}, we plot the	implied volatility smile for S\&P 500 on the $29^{th}$ of November 2012 (time to maturity of 22 days on the left and of 9 months and 22 days on the right) and on the $29^{th}$ of November 2013 (time to maturity of 21 days on the left and of 9 months and 21 days on the right). The NIG ATS process, Sato process, and LTS process implied volatility are plotted together with the market-implied volatility. 
	As in the case of the $30^{th}$ of May 2013 (cf. Figure \ref{figure:vol}) the ATS reproduces the smiles significantly better than the alternatives, the improvement is particularly evident for small maturities.
	\begin{center}
		\begin{minipage}[t]{1\textwidth}
			\centering
			\includegraphics[width=0.9\textwidth]{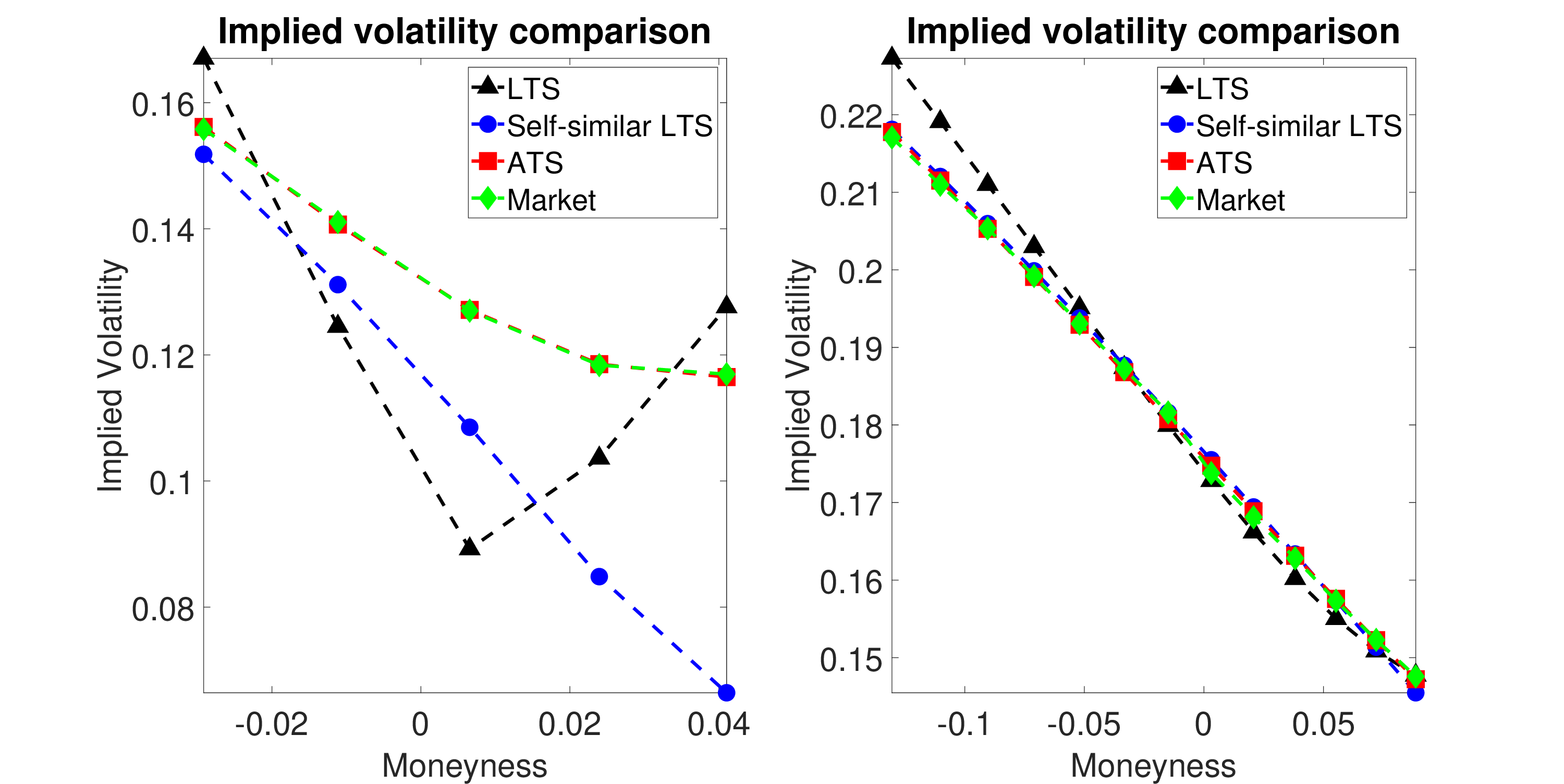} 
			\captionof{figure}{\small 
				Implied volatility smile for S\&P 500 at the $29^{th}$ of November 2012 (6 months before) for a given time to maturity: 22 days (on the left) and 9 months and 22 days (on the right). 
				The NIG ATS process, Sato process, and LTS process implied volatility are plotted together with the market-implied volatility. 
				ATS reproduces the smile significantly better than the alternatives, the improvement is particularly evident for small maturities.}\label{figure:volMarch}
		\end{minipage}
	\end{center}
	\begin{center}
		\begin{minipage}[t]{1\textwidth}
			\centering
			\includegraphics[width=0.9\textwidth]{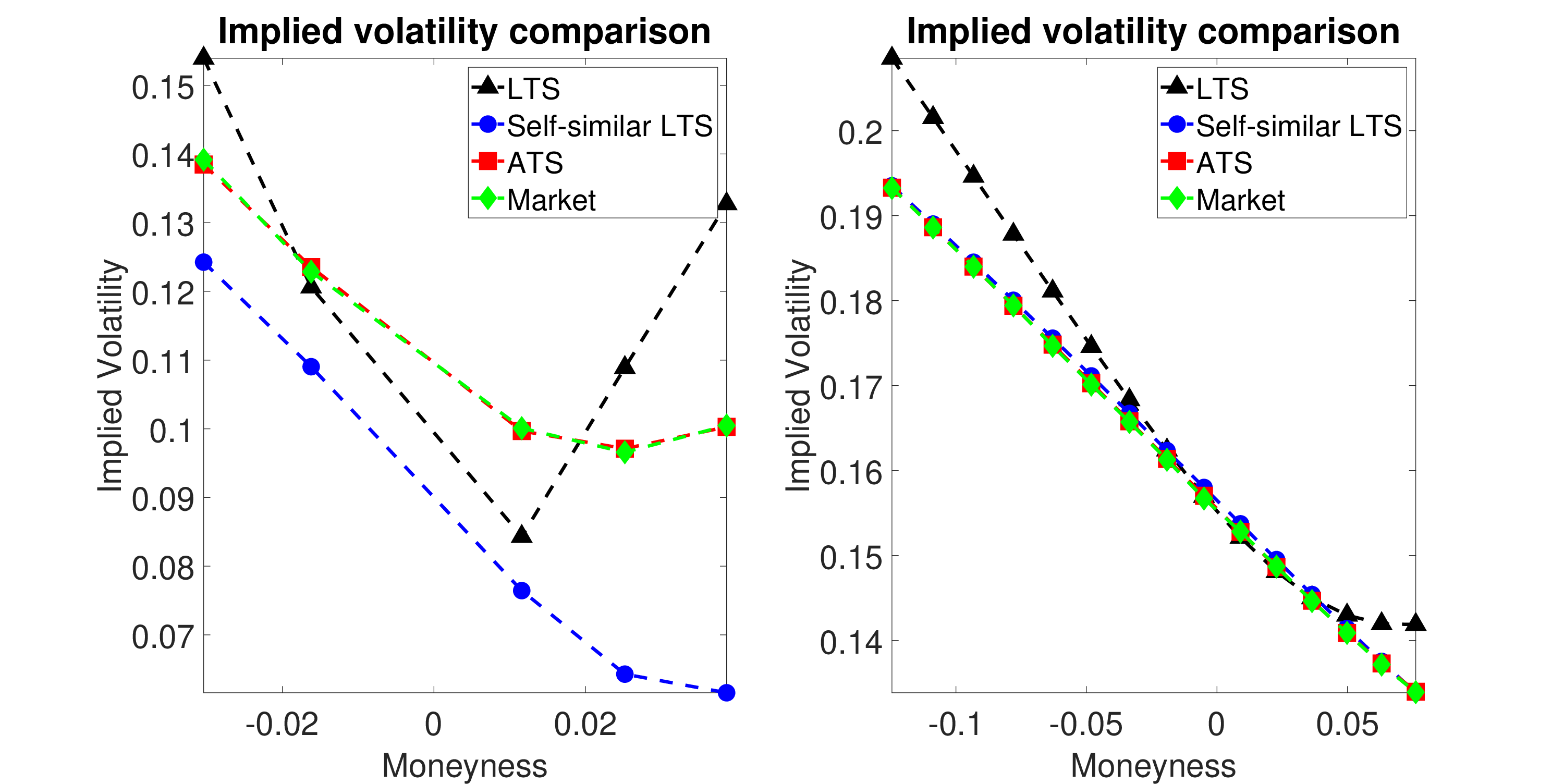} 
			\captionof{figure}{\small 
				Implied volatility smile for S\&P 500 at the $29^{th}$ of November 2013 (6 months after) for a given time to maturity: 21 days (on the left) and 9 months and 21 days (on the right). 
				The NIG ATS process, Sato process, and LTS process implied volatility are plotted together with the market-implied volatility. 
				ATS reproduces the smile significantly better than the alternatives, the improvement is particularly evident for small maturities.}\label{figure:volJuly}
		\end{minipage}
	\end{center}
	In Figure \ref{figure:skew_comparison_dates}, we plot the market and the ATS NIG implied volatility \textit{skew}  for EURO STOXX 50 w.r.t. the times to maturity on the $27^{th}$ of February and on the $30^{th}$ of August 2013.  In both cases, the calibrated ATS replicates accurately the market implied volatility \textit{skew}, as already observed in Section \ref{section:calibration} for the $30^{th}$ of May 2013. Similar results hold for the other two dates and in the S\&P 500 case.
		\begin{center}
		\begin{minipage}[t]{1\textwidth}
			\centering
			\includegraphics[width=0.9\textwidth]{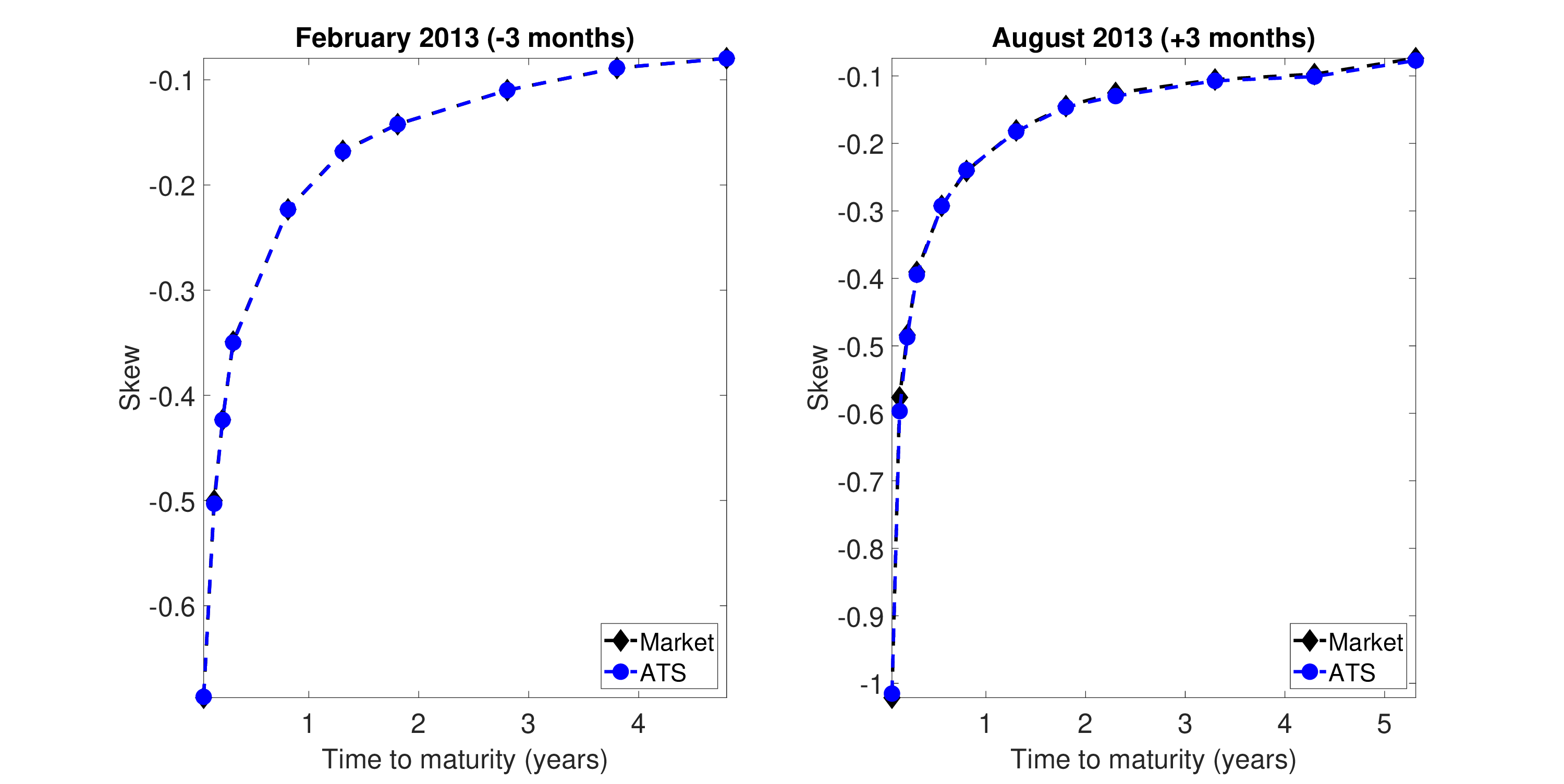} 
			\captionof{figure}{\small  The market and the ATS NIG implied volatility \textit{skew}  for EURO STOXX 50 w.r.t. the times to maturity on the $27^{th}$ of February and on the $30^{th}$ of August 2013. Again, ATS replicates the market implied implied volatility \textit{skew}.
			}\label{figure:skew_comparison_dates}
		\end{minipage}
	\end{center}
	In Figures \ref{figure:SPPNIG_comparison} and \ref{figure:EUVG_comparison}, we plot the weighted regression lines and the observed time-dependent parameters $\ln {\hat k}_\theta$ and $\ln {\hat \eta}_\theta$ 
	with their confidence intervals 
	for S\&P $500$ and EURO STOXX $50$ on the $29^{th}$ of November 2012 (on the top) and on the $29^{th}$ of November 2013 (on the bottom).\footnote{Results for the $27^{th}$ of February 2013 and the $30^{th}$ of August 2013 are available upon request.}
	The confidence intervals are two times the standard deviations of $\ln {\hat k}_\theta$,  of
	$\ln{\hat \eta}_\theta$ and of $\ln \theta$. In both days, the observed scalings of $\hat{k}_\theta$ and $\hat{\eta}_\theta$ are equivalent to the ones observed in  Figures \ref{figure:SPPNIG} and \ref{figure:EUVG}.\\
	\begin{center}
		\begin{minipage}[t]{1\textwidth}
			\centering
			\centerline{\includegraphics[width=1.15\textwidth]{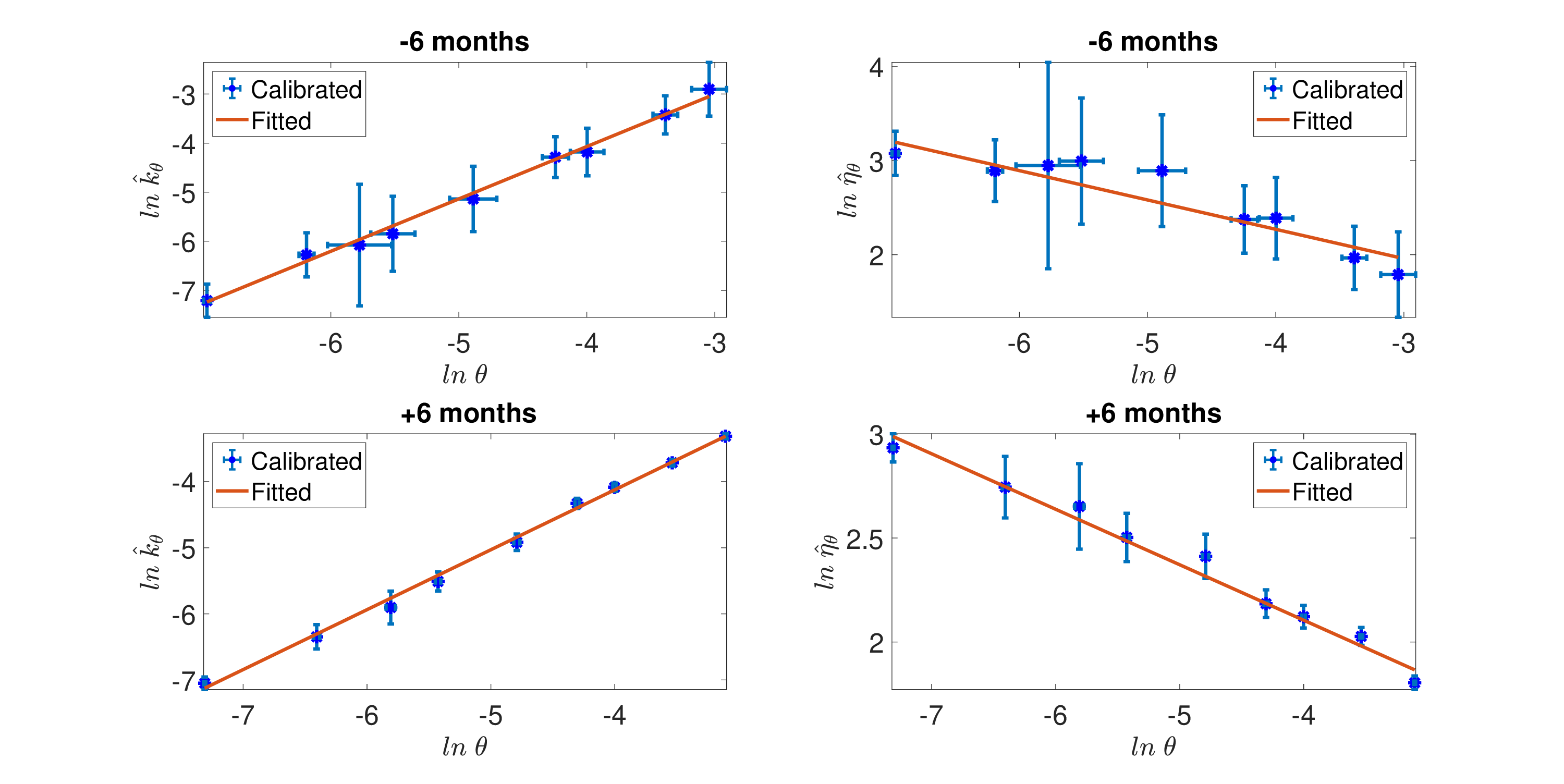}} 
			\captionof{figure}{\small Weighted regression line and the observed time-dependent parameters $\ln {\hat k}_\theta$ and $\ln {\hat \eta}_\theta$ w.r.t. $\ln\theta$ for the NIG calibrated model for S\&P 500 on   the $29^{th}$ of November 2012 (on the top) and on the $29^{th}$ of November 2013 (on the bottom). We plot a confidence interval equal to two times the corresponding standard deviation. Notice that, also in this case,  confidence intervals on $\ln {\hat k}_\theta$  and $\ln {\hat \eta}_\theta$  are one order of magnitude wider than confidence intervals on $\ln \theta$. The observed scalings of $\hat{k}_\theta$ and $\hat{\eta}_\theta$ are equivalents to the ones observed in  Figures \ref{figure:SPPNIG} and \ref{figure:EUVG}. The values of $\theta$ correspond to times to maturity that goes from tree weeks to two years and a half (all available maturities on Reuters Datastream  dataset).}\label{figure:SPPNIG_comparison}
		\end{minipage}\\ \bigskip
		\begin{minipage}[t]{1\textwidth}
			\centering
			\centerline{\includegraphics[width=1.15\textwidth]{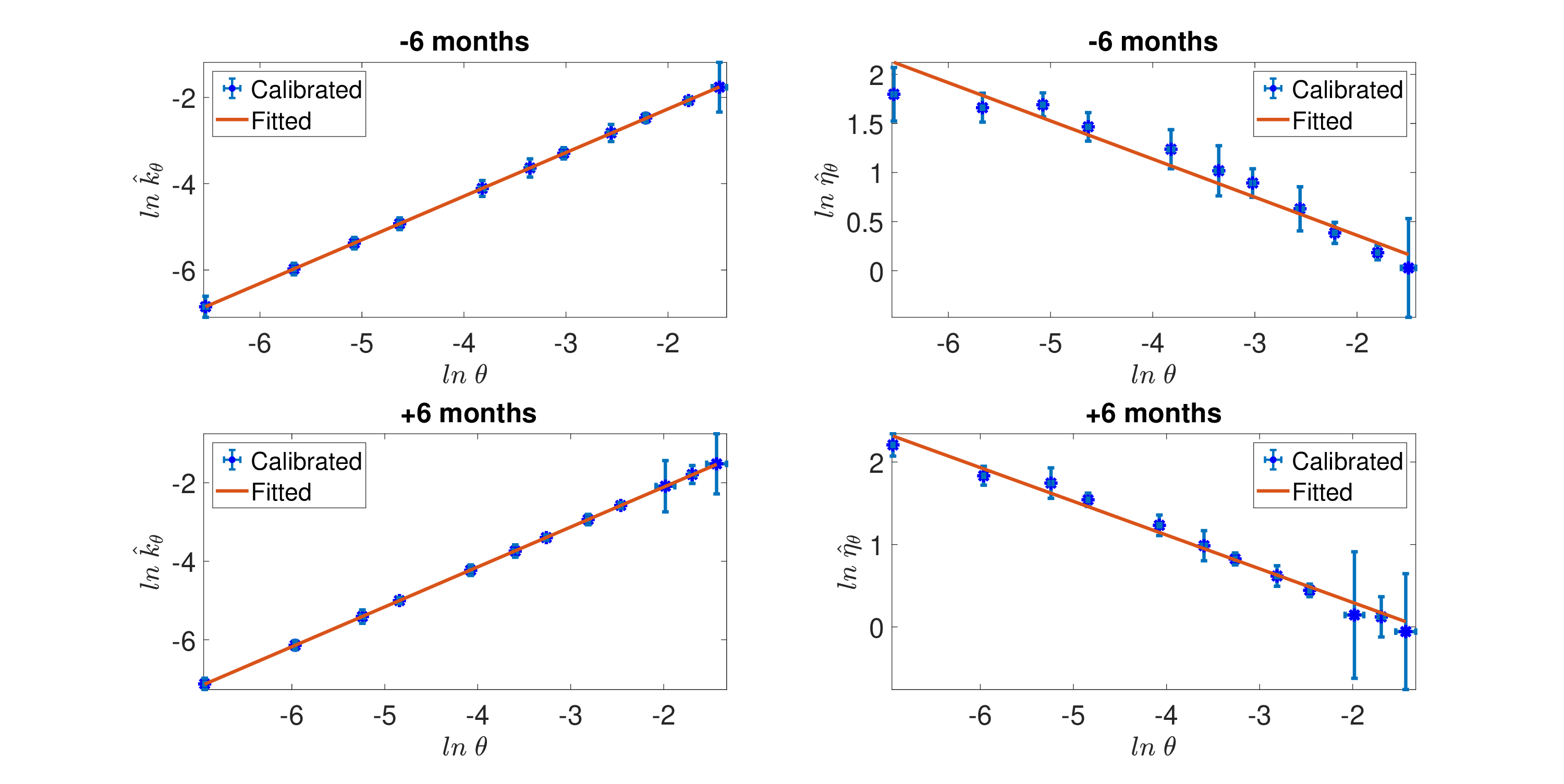}}
			\captionof{figure}{\small	\small Weighted regression line and the observed time-dependent parameters $\ln {\hat k}_\theta$ and $\ln {\hat \eta}_\theta$ w.r.t. $\ln\theta$ for the VG calibrated model for EURO STOXX 50 on   the $29^{th}$ of November 2012 (on the top) and on the $30^{th}$ of November 2013 (on the bottom). We plot a confidence interval equal to two times the corresponding standard deviation. Notice that, also in this case,  confidence intervals on $\ln {\hat k}_\theta$  and $\ln {\hat \eta}_\theta$  are one order of magnitude wider than confidence intervals on $\ln \theta$. The observed scalings of $\hat{k}_\theta$ and $\hat{\eta}_\theta$ are equivalents to the ones observed in  Figures \ref{figure:SPPNIG} and \ref{figure:EUVG}.  The values of $\theta$ correspond to times to maturity that goes from tree weeks to five years (all available maturities on Reuters Datastream dataset).}\label{figure:EUVG_comparison}
		\end{minipage}\\ \bigskip	
	\end{center}
	The results presented in Table \ref{tab:calibration_comparison} and in Figures 9-14 are equivalent to the one of Section \ref{section:calibration}. This analysis confirms the robustness, over a one-year time interval, of the excellent calibration performances and the power scaling behavior of the ATS.
	\section{Conclusions}
	\label{section:conclusions}
	
	In this paper, we introduce a new broad family of stochastic processes that we call additive normal tempered stable processes (ATS).
	An interesting subcase of ATS presents a power-law scaling of the time-dependent parameters.
	
	\bigskip
	
	We have considered all quoted options on S\&P 500 and EURO STOXX 50 at 11:00 am New York Time on the $30^{th}$ of May 2013.
	The dataset considers options with a time to maturity starting from three weeks and up to several years.
	We calibrate the ATS processes on the options of both indices,
	showing that ATS  processes present better calibration features than LTS and Sato processes.
	The observed improvement of ATS is even of two orders of magnitude in terms of MSE, as reported in Table \ref{tab: MSE_APE}. ATS replicates accurately market implied volatility term structure and \textit{skew} as observed in Figures \ref{figure:vol} and \ref{figure:skew}.
	
	The quality of ATS calibration results looks quite incredible.
	In Subsection {\bf \ref{section:scaling}},
	we have shown that once the volatility term structure has been taken into account, 
	the whole implied volatility surface is calibrated accurately with only two free parameters.

	We also construct a re-scaled ATS process via a time-change based on the implied volatility term structure.
	We show that the re-scaled process calibrated parameters exhibit a power-law behavior. 
	The statistical relevance of the scaling properties is discussed in detail. 
\bigskip
	
	We have compared some model characteristics with the two alternative additive processes present in the financial literature.
	These two classes fail to reproduce some stylized facts observed in market data, which are adequately described by ATS processes.\\
	Finally, we have verified the robustness, over a one-year time interval, of the excellent calibration performances and the power scaling behavior of the ATS.\\

	As for future research, two main promising directions appear evident. 
	First, it can be interesting to extend ATS processes to the commodity asset class, in general, and to the oil \citep{shiraya2011pricing,kyriakou2016jumps} 
	and freight markets \citep{prokopczuk2011pricing,nomikos2013freight}, in particular;
this model extension should allow for mean-reversion and seasonality patterns in prices, which are typically found in empirical studies \citep[see, e.g.][and references therein]{benth2014quantitative}.
	Second,  it would be worthy to develope a fast and reliable simulation algorithm for ATS processes \citep[see, e.g.][]{azzone2021additive}
	and to study pricing techniques for exotic derivatives 
	\citep[e.g., generalizing the techniques for path-dependent exotics products, as Asian options, in][through the characteristic function of the ATS increments]{fusai2008pricing,vcerny2011improved,fusai2016general}.

	\section*{Acknowledgements}
	We thank P. Carr, J. Gatheral, and F. Benth for enlightening discussions on this topic.  We thank also G. Guatteri, M.P. Gregoratti, J. Guyon, P. Spreij, and all participants to WSMF 2019 in Lunteren, to VCMF 2019 in Vienna, and to AFM 2020 in Paris. We are grateful to the Editor and the Referees for their useful comments and careful review. 
	R.B. feels indebted to P. Laurence for several helpful and wise suggestions on the subject.

	\bibliography{sources}

\begin{thebibliography}{36}
\providecommand{\natexlab}[1]{#1}
\providecommand{\selectlanguage}[1]{\relax}

\bibitem[{Abramowitz and Stegun(1948)}]{abramowitz1948handbook}
Abramowitz, M. and Stegun, I.A., 1948. \textit{Handbook of mathematical
  functions with formulas, graphs, and mathematical tables}, vol.~55, US
  Government printing office.

\bibitem[{Azzone and Baviera(2021)}]{azzone2021additive}
Azzone, M. and Baviera, R., 2021. Simulation of additive normal tempered stable
  processes, \textit{In preparation}.

\bibitem[{Ballotta and Ray{\'e}e(2018)}]{ballotta2018smiles}
Ballotta, L. and Ray{\'e}e, G., 2018. Smiles \& smirks: a tale of factors,
  \textit{Available at SSRN 2980349}.

\bibitem[{Barndorff-Nielsen(1997)}]{barndorff1997normal}
Barndorff-Nielsen, O.E., 1997. Normal inverse {G}aussian distributions and
  stochastic volatility modelling, \textit{Scandinavian Journal of Statistics},
  24~(1), 1--13.

\bibitem[{Baviera(2007)}]{baviera2007}
Baviera, R., 2007. Gigi model: a simple stochastic volatility approach for
  multifactor interest rates, \textit{Available at SSRN 977479}.

\bibitem[{Benth \textit{et~al.}(2014)Benth, Kholodnyi, and
  Laurence}]{benth2014quantitative}
Benth, F.E., Kholodnyi, V.A., and Laurence, P., 2014. Quantitative energy
  finance, \textit{Modelling, pricing, and hedging in energy and commodity
  markets, Springer}.

\bibitem[{Benth and Sgarra(2012)}]{benth2012risk}
Benth, F.E. and Sgarra, C., 2012. The risk premium and the {E}sscher transform
  in power markets, \textit{Stochastic Analysis and Applications}, 30~(1),
  20--43.

\bibitem[{Carr \textit{et~al.}(2007)Carr, Geman, Madan, and Yor}]{carr2007self}
Carr, P., Geman, H., Madan, D.B., and Yor, M., 2007. Self-decomposability and
  option pricing, \textit{Mathematical {F}inance}, 17~(1), 31--57.

\bibitem[{Carr and Madan(1999)}]{carr1999option}
Carr, P. and Madan, D., 1999. Option valuation using the fast {F}ourier
  transform, \textit{Journal of Computational Finance}, 2~(4), 61--73.

\bibitem[{{\v{C}}ern{\`y} and Kyriakou(2011)}]{vcerny2011improved}
{\v{C}}ern{\`y}, A. and Kyriakou, I., 2011. An improved convolution algorithm
  for discretely sampled {A}sian options, \textit{Quantitative Finance},
  11~(3), 381--389.

\bibitem[{Cont and Tankov(2003)}]{Cont}
Cont, R. and Tankov, P., 2003. \textit{Financial Modelling with jump
  processes}, Chapman and Hall/CRC Financial Mathematics Series.

\bibitem[{Eberlein and Madan(2009)}]{eberlein2009sato}
Eberlein, E. and Madan, D.B., 2009. Sato processes and the valuation of
  structured products, \textit{Quantitative Finance}, 9~(1), 27--42.

\bibitem[{Fusai and Kyriakou(2016)}]{fusai2016general}
Fusai, G. and Kyriakou, I., 2016. General optimized lower and upper bounds for
  discrete and continuous arithmetic {A}sian options, \textit{Mathematics of
  Operations Research}, 41~(2), 531--559.

\bibitem[{Fusai and Meucci(2008)}]{fusai2008pricing}
Fusai, G. and Meucci, A., 2008. Pricing discretely monitored {A}sian options
  under l{\'e}vy processes, \textit{Journal of Banking \& Finance}, 32~(10),
  2076--2088.

\bibitem[{Gatheral(2011)}]{gatheral2011volatility}
Gatheral, J., 2011. \textit{The volatility surface: a practitioner's guide},
  vol. 357, John Wiley \& Sons.

\bibitem[{George \textit{et~al.}(1991)George, Kaul, and
  Nimalendran}]{george1991estimation}
George, T.J., Kaul, G., and Nimalendran, M., 1991. Estimation of the bid--ask
  spread and its components: A new approach, \textit{The Review of Financial
  Studies}, 4~(4), 623--656.

\bibitem[{Konikov and Madan(2002)}]{konikov2002option}
Konikov, M. and Madan, D.B., 2002. Option pricing using variance gamma {M}arkov
  chains, \textit{Review of Derivatives Research}, 5~(1), 81--115.

\bibitem[{Kyriakou \textit{et~al.}(2016)Kyriakou, Pouliasis, and
  Papapostolou}]{kyriakou2016jumps}
Kyriakou, I., Pouliasis, P.K., and Papapostolou, N.C., 2016. Jumps and
  stochastic volatility in crude oil prices and advances in average option
  pricing, \textit{Quantitative Finance}, 16~(12), 1859--1873.

\bibitem[{Lewis(2001)}]{lewis2001simple}
Lewis, A.L., 2001. A simple option formula for general jump-diffusion and other
  exponential {L}{\'e}vy processes, \textit{Mimeo}.

\bibitem[{Li \textit{et~al.}(2016)Li, Li, and Mendoza-Arriaga}]{li2016Additive}
Li, J., Li, L., and Mendoza-Arriaga, R., 2016. Additive subordination and its
  applications in finance, \textit{Finance and Stochastics}, 20~(3), 589--634.

\bibitem[{Li and Linetsky(2014)}]{li2014time}
Li, L. and Linetsky, V., 2014. Time-changed ornstein--uhlenbeck processes and
  their applications in commodity derivative models, \textit{Mathematical
  Finance}, 24~(2), 289--330.

\bibitem[{Madan \textit{et~al.}(1998)Madan, Carr, and
  Chang}]{madan1998variance}
Madan, D.B., Carr, P.P., and Chang, E.C., 1998. The variance gamma process and
  option pricing, \textit{Review of Finance}, 2~(1), 79--105.

\bibitem[{Madan and Seneta(1990)}]{MadanSeneta1990}
Madan, D.B. and Seneta, E., 1990. The variance gamma ({VG}) model for share
  market returns, \textit{Journal of {B}usiness}, 511--524.

\bibitem[{Nomikos \textit{et~al.}(2013)Nomikos, Kyriakou, Papapostolou, and
  Pouliasis}]{nomikos2013freight}
Nomikos, N.K., Kyriakou, I., Papapostolou, N.C., and Pouliasis, P.K., 2013.
  Freight options: Price modelling and empirical analysis,
  \textit{Transportation Research Part E: Logistics and Transportation Review},
  51, 82--94.

\bibitem[{Ornthanalai(2014)}]{ornthanalai2014levy}
Ornthanalai, C., 2014. Levy jump risk: Evidence from options and returns,
  \textit{Journal of Financial Economics}, 112~(1), 69--90.

\bibitem[{Petersen and Fialkowski(1994)}]{petersen1994posted}
Petersen, M.A. and Fialkowski, D., 1994. Posted versus effective spreads: Good
  prices or bad quotes?, \textit{Journal of Financial Economics}, 35~(3),
  269--292.

\bibitem[{Prokopczuk(2011)}]{prokopczuk2011pricing}
Prokopczuk, M., 2011. Pricing and hedging in the freight futures market,
  \textit{Journal of Futures Markets}, 31~(5), 440--464.

\bibitem[{Roll(1984)}]{roll1984simple}
Roll, R., 1984. A simple implicit measure of the effective bid-ask spread in an
  efficient market, \textit{The Journal of {F}inance}, 39~(4), 1127--1139.

\bibitem[{Ryan(2008)}]{ryan2008modern}
Ryan, T.P., 2008. \textit{Modern regression methods}, vol. 655, John Wiley \&
  Sons.

\bibitem[{Sato(1991)}]{sato1991self}
Sato, K.I., 1991. Self-similar processes with independent increments,
  \textit{Probability Theory and Related Fields}, 89~(3), 285--300.

\bibitem[{Sato(1999)}]{Sato}
Sato, K.I., 1999. \textit{L\'evy processes and infinitely divisible
  distributions}, Cambridge University Press.

\bibitem[{Schoutens(2003)}]{Schoutens}
Schoutens, W., 2003. \textit{L{\'e}vy processes in finance}, Wiley.

\bibitem[{Seber and Wild(1989)}]{seber1989nonlinear}
Seber, G. and Wild, C., 1989. \textit{Nonlinear regression, 768 pp}, Wiley, New
  York.

\bibitem[{Shiraya and Takahashi(2011)}]{shiraya2011pricing}
Shiraya, K. and Takahashi, A., 2011. Pricing average options on commodities,
  \textit{Journal of Futures Markets}, 31~(5), 407--439.

\bibitem[{Taylor(1997)}]{Taylor1997}
Taylor, J., 1997. \textit{Introduction to error analysis, the study of
  uncertainties in physical measurements}, University Science Books, Sausalito,
  CA.

\bibitem[{York(1968)}]{york1968least}
York, D., 1968. Least squares fitting of a straight line with correlated
  errors, \textit{Earth and Planetary Science Letters}, 5, 320--324.

\end{thebibliography}
	\bibliographystyle{tandfx}
	
	\bigskip
	\bigskip
	\bigskip
	\bigskip

	\section*{Notation}
	
	\begin{center}

		\begin{tabular} {|c|l|}
			\toprule
			\textbf{Symbol}& \textbf{Description}\\ \bottomrule
			$A_t$& diffusion term of the additive process $\left\{X_t \right\}_{t\geq 0}$\\
			$B_T$ & discount factor between value date and $T$ \\
			$\mathbb{B}(\mathbb{R})$ & Borel sigma algebra on $\mathbb{R}$ \\ 
			$C\left(T, K\right)$ & call option price at value date with maturity $T$ and strike $K$\\
			$\left\{f_t\right\}_{t\geq 0}$& ATS process that models the forward exponent \\
			$\left\{\hat{f}_\theta\right\}_{\theta\geq 0}$& re-scaled ATS process w.r.t. the time $\theta=\sigma^2_T \; T$\\
			$F_t (T)$&  price  at time $t$ of a Forward contract with maturity $T$\\
			$k$ & variance of jumps of LTS \\
			$k_t$ &  variance of jumps of ATS\\
			$\hat{k}_\theta$ & re-scaled  variance of jumps of ATS\\
			$\bar{k}$ & constant part of the re-scaled  variance of jumps of ATS $\hat{k}_\theta$\\
			${\cal L}_t$& Laplace transform of the subordinator $S_t$ in (\ref{eq:lap_transf})\\ 
			$\left\{S_t\right\}_{t\geq 0}$& L\'evy subordinator \\
			$T$& option time to maturity\\
			$W_t$ & Brownian motion\\
			$x$&  option moneyness\\
			$ \alpha$ & \textit{Index of stability}: tempered stable parameter of ATS, $\alpha\in [0,1)$]\\
			$ \beta$& scaling parameter of $\hat{k}_\theta$\\
			$\gamma_t$ & drift term of additive process $\left\{X_t \right\}_{t\geq 0}$\\ 
			$\Gamma(*)$ & Gamma function evaluated in $*$\\
			$\delta$ & scaling parameter of $\hat{\eta}_\theta$\\
			$\varphi$&  deterministic drift term of  LTS\\
			$\varphi_t$&  deterministic drift term of ATS \\
			$\phi^c  $ &  characteristic function of the forward exponent\\
			$\eta$ &   \textit{skew} parameter of LTS \\
			$\eta_t$ &  \textit{skew} parameter of ATS\\
			$\hat{\eta}_\theta$ & re-scaled  \textit{skew} parameter of ATS \\
			$\bar{\eta}$ & constant part of the  re-scaled ATS \textit{skew} parameter $\hat{\eta}_\theta$\\
			$\nu_t$ & L\'evy measure of the additive process $\left\{X_t \right\}_{t\geq 0}$\\
			$\sigma$ &  diffusion parameter of LTS \\
			$\sigma_t$ &   diffusion parameter of ATS\\
			$\hat{\sigma}_\theta$ &re-scaled diffusion parameter of ATS, equal to one\\
			$\bar{ \sigma}$ &  constant diffusion parameter of ATS\\
			$\theta $ & re-scaled maturity, defined as $\sigma^2_T \; T$ \\
			\bottomrule
		\end{tabular}
		\thispagestyle{plain}
	\end{center}	
	
	\clearpage
	\begin{appendices}
		\section{Proofs}

		We  start proving a technical \textbf{Lemma} that we use in the proof of \textbf{Theorem \ref{theorem:f_Additive}}.
		
		\begin{lemma}\label{lemma:gamma_t}$ $\\
			If $ \lim_{t\to 0} t\,\sigma^2_t\,\eta_t =0$, then
			\[\lim_{t\to 0} \mathop{\mathlarger{\int}}^\infty_0{ds\displaystyle \frac{t}{\Gamma(1-\alpha)}\left({\frac{1-\alpha}{k_t}}\right)^{1-\alpha}\left(\frac{e^{-\left(1-\alpha\right) \; s/ k_t }}{s^{1+\alpha}}\right)}\mathop{\mathlarger{\int}}_{|x|<1}dx\frac{x}{\sqrt{2 \pi s}{\sigma}_t} e^{-\left(\frac{x+s{\sigma}_t^2\left(\eta_t+1/2\right)}{\sqrt{s}{\sigma}_t}\right)^2}=0\;\;.\] 
		\end{lemma}
		\begin{proof}
			
			\begin{align*}
			\\
			&\left|\mathop{\mathlarger{\int}}_{|x|<1}dx\frac{x}{\sqrt{2 \pi s}{\sigma}_t} e^{-\left(\frac{x+s{\sigma}_t^2\left(\eta_t+1/2\right)}{\sqrt{s}{\sigma}_t}\right)^2}\right|\\ \leq &
			\left|\mathop{\mathlarger{\int}}_{|x|<1}dx\frac{x}{\sqrt{2 \pi s}{\sigma}_t} e^{-\left(\frac{x+s{\sigma}_t^2\left(\eta_t+1/2\right)}{\sqrt{s}{\sigma}_t}\right)^2}   +\mathop{\mathlarger{\int}}_{1}^\infty dx\frac{x}{\sqrt{2 \pi s}{\sigma}_t} \left(e^{-\left(\frac{x+s{\sigma}_t^2\left(\eta_t+1/2\right)}{\sqrt{s}{\sigma}_t}\right)^2}-e^{-\left(\frac{-x+s{\sigma}_t^2\left(\eta_t+1/2\right)}{\sqrt{s}{\sigma}_t}\right)^2}\right)  
			\right|\\
			=& \sigma_t^2 s\left|\frac{1}{2}+\eta_t\right|\;\;.
			\end{align*}
			The inequality is due to the fact that both terms inside the right-hand side absolute value are positive if $(1/2+\eta_t)$ is positive and are negative if $(1/2+\eta_t)$ is negative.  
			Now it is possible to write the following bound
			\begin{align*}
			&\left|\mathop{\mathlarger{\int}}^\infty_0{ds\displaystyle \frac{t}{\Gamma(1-\alpha)}\left({\frac{1-\alpha}{k_t}}\right)^{1-\alpha}\left(\frac{e^{-\left(1-\alpha\right) \; s/ k_t }}{s^{1+\alpha}}\right)}\mathop{\mathlarger{\int}}_{|x|<1}dx\frac{x}{\sqrt{2 \pi s}{\sigma}_t} e^{-\left(\frac{x+s{\sigma}_t^2\left(\eta_t+1/2\right)}{\sqrt{s}{\sigma}_t}\right)^2}\right|\\
			&\leq  \sigma_t^2 \left|\frac{1}{2}+\eta_t\right| \mathop{\mathlarger{\int}}^\infty_0{ds\displaystyle \frac{t}{\Gamma(1-\alpha)}\left({\frac{1-\alpha}{k_t}}\right)^{1-\alpha}\left(\frac{e^{-\left(1-\alpha\right) \; s/ k_t }}{s^{\alpha}}\right)}=t\sigma_t^2 \left|\frac{1}{2}+\eta_t\right|\;\;,
			\end{align*}
			where the last equality is due to the definition of $\Gamma(1-\alpha)$.
			We prove the thesis by the squeeze theorem
		\end{proof}
		
		\paragraph{Proof of Theorem \ref{theorem:f_Additive}}$ $\\
		The idea of this proof is 	to show that there exists an additive process with the characteristic function  in (\ref{laplace})  using the result in
		\citet[][ Th.9.8, p.52]{Sato}.

		At any given time $t>0$ the characteristic function in (\ref{laplace}) is the characteristic function of a LTS (\ref{laplace_levy}), at time $t$, with parameters ${k}=k_t$, ${ \eta}=\eta_t$, ${\sigma}=\sigma_t$ and $\varphi =\varphi_t$. Hence,  we have an expression for the generating triplet of (\ref{laplace})  \citep[see, e.g.][eq. 4.24, p.130]{Cont}\[
		\begin{cases}
		A_t&=0\\
		\gamma_t &=\mathop{\mathlarger{\int}}^\infty_0{ds\displaystyle \frac{t}{\Gamma(1-\alpha)}\left({\frac{1-\alpha}{k_t}}\right)^{1-\alpha}\left(\frac{e^{-\left(1-\alpha\right) \; s/ k_t }}{s^{1+\alpha}}\right)}\mathop{\mathlarger{\int}}_{|x|<1}dx\dfrac{x}{\sqrt{2 \pi s}{\sigma}_t} e^{-\left(\frac{x+s{\sigma}_t^2\left(\eta_t+1/2\right)}{\sqrt{s}{\sigma}_t}\right)^2}+t\varphi_t \\
		\nu_t(x)&=\dfrac{tC\left(\alpha,k_t,{\sigma}_t,\eta_t\right)}{|x|^{1/2+\alpha}}e^{-(\eta_t+1/2)x}K_{\alpha+1/2}\left(|x|\sqrt{{\left(1/2+\eta_t\right)^2+2(1-\alpha)/( k_t\,{\sigma}^2_t)}}\right)
		
		\end{cases}\;\;,
		\]
		with \[C\left(\alpha,k_t,{\sigma}_t,\eta_t\right):=\frac{2}{\Gamma(1-\alpha) \sqrt{2 \pi}}\left(\frac{1-\alpha}{k_t}\right)^{1-\alpha}{\sigma}^{2\alpha}_t\left(\left(1/2+\eta_t\right)^2+2(1-\alpha)/(k_t\, {\sigma}^2_t)\right)^{\alpha/2+1/4}\;\;,\]
		and 
		\[ K_\nu(z):=\frac{e^{-z}}{\Gamma\left(\nu+\frac{1}{2}\right)}\sqrt{\frac{\pi}{2 z}}\int_0^\infty e^{-s}s^{\nu-1/2}\left(\frac{s}{2z}+1\right)^{\nu-1/2}ds\;\;\]
		is the modified Bessel function of the second kind  \citep[see, e.g.][Ch.9 p.376]{abramowitz1948handbook}.
		For $t=0$, as usual in additive processes, we set $\gamma_0=0$, $A_0=0$ and $\nu_0=0$.\\ 
		\bigskip\\
		First, we verify that $\nu_t(x)$ is a non decreasing function of $t$.
		It is possible to identify two expressions in the jump measure $\nu_t(x)$ \begin{align}
		&e^{-x(1/2+\eta_t)-|x|\sqrt{\left(1/2+\eta_t\right)^2+ 2(1-\alpha)/( \sigma_t^2k_t)}}\label{eq:exponential} \\
		&\frac{ t^{1/\alpha}\sigma^2_t}{k_t^{(1-\alpha)/\alpha}}\left(\frac{s }{|x|}+\sqrt{\left(1/2+\eta_t\right)^2+2(1-\alpha)/(\sigma_t^2k_t)}\right)\label{eq:power}	\,\;.
		\end{align} We point out that expression (\ref{eq:power}) is inside the integral and depends on the integration variable $s\geq 0$. If these two expressions, (\ref{eq:exponential}) and (\ref{eq:power}), are non decreasing w.r.t. $t$ for any $t$, $x$ and $s\geq0$ then the jump measure is non decreasing. Expression (6) is non decreasing because $g_1$ and $g_2$ are non decreasing by hypothesis 1. 
		Hypothesis 1 on $g_1$ and $g_2$ also implied that the squared root in (\ref{eq:power}) is non increasing for any $t$ and then, because condition 1 on $g_3$ holds, the prefactor $\frac{ t^{1/\alpha}\sigma^2_t}{k_t^{(1-\alpha)/\alpha}}$ is non decreasing (even multiplied by $s/|x|$).
		Thus, (\ref{eq:power}) is non decreasing for any $t$, $x$ and $s\geq0$ because it is the sum of a non decreasing function and $g_3$, a non decreasing function by hypothesis. This proves that $\nu_t(x)$ is non decreasing in $t$.\\
		\bigskip\\
		Second, we prove that $\lim_{t\to 0}\nu_t(x)=0$ for $x\neq 0$; this is equivalent to demonstrate that   (\ref{eq:exponential}) or  (\ref{eq:power}) go to zero as $t$ goes to zero.
		We show that this happens in all possible cases. We first consider the case where
		\begin{equation}
		\label{eq:limit}
		\lim_{t\to0}k_t >0 \;\;\;\;\text{and}\;\;\;\;	\lim_{t\to0}|(1/2+\eta_t)| <\infty\;\;.
		\end{equation}
		In this case is evident that expression (\ref{eq:power}) goes to zero for small $t$. Otherwise, when (\ref{eq:limit}) is not true, we have to distinguish two further cases depending on whether \begin{equation}\label{eq:limit2}
		\lim_{t\to0}(1/2+\eta_t)\,k_t\sigma_t^2=0
		\end{equation}
		holds. If (\ref{eq:limit2}) is true expression (\ref{eq:exponential}) goes to zero, otherwise, because of condition 2  on $t\,\eta_t^\alpha\,\sigma_t^{2\alpha}/k_t^{1-\alpha}$, expression (\ref{eq:power}) goes to zero.
		This proves that $\lim_{t\to0}\nu_t(x)=0$ for any $x\neq 0$.\\
		
		We can now check whether the triplet satisfies the conditions in 	\citet[][ Th.9.8, p.52]{Sato}.

		\begin{enumerate}
			\item The triplet has no diffusion term.
			\item ${	{ \nu}}_t$ is not decreasing in $t$.
			\item The continuity of ${ 	{\nu}}_t(B)$ and $\gamma_t$ , where $B\in\mathbb{B}\left(\mathbb{R}^+\right)$ and $B \subset \{ x : |x|>\epsilon >0 \}$, is obvious for $t>0$: it is a natural consequence of the composition of continuous functions.
			For $t=0$ we have to prove that the limits of
			${ 	{\nu}}_t(B)$ and ${ \gamma}_t$ are 0. We have already proven that 	${ {\nu}}_t(x)$ is  non decreasing in $t$ and that $\lim_{t_\to 0}\nu_t(x)=0$, $\forall x \neq 0$. The convergence of ${ 	{\nu}}_t(B)$ to 0 is due to the dominated convergence theorem. The convergence of ${ \gamma}_t$ is because of \textbf{Lemma \ref{lemma:gamma_t}} and because $t\varphi_t$ goes to zero by definition of ATS. 
			\qed
		\end{enumerate}

		\paragraph{Proof of Proposition \ref{th6}}$ $
		
		A forward contract, valued in $t$ with delivery in $T$, is
		$
		F_t\left(T\right)=F_0\left(T\right) \; e^{f_t} 
		$, also for an ATS, as in (\ref{eq:FwdBasic}) for the LTS. 
		
		\smallskip 
		
		Let us prove the sufficient condition.
		If the forward is martingale
		\begin{equation*}
		\mathbb{E}\left[F_t(T)\middle \vert \mathcal{F}_0 \right]=F_0\left(T\right)\;\; .
		\end{equation*}
		This is equivalent to impose that
		\begin{equation}
		\mathbb{E}\left[e^{f_t}\middle \vert \mathcal{F}_0\right] = 1\; ,
		\label{eq:MartCond}
		\end{equation}
		that is, the characteristic function of $f_t$ computed in $-i$ is equal to one.
		From equation (\ref{laplace})
		\begin{equation}
		\mathbb{E}[e^{f_t} \vert \mathcal{F}_0 ]={\cal L}_t\left(\left(\eta_t+\frac{1}{2}\right)\sigma_t^2-\frac{\sigma_t^2}{2};\;k_t,\;\alpha\right)e^{\varphi_t t}={\cal L}_t\left(\sigma_t^2\eta_t;\;k_t,\;\alpha\right)e^{\varphi_t t}\;\;.
		\label{eq:MartCond2}
		\end{equation} 
		Imposing the condition (\ref{eq:MartCond}), we get $\varphi_t$.
		
		\smallskip 
		
		Let us prove the necessary condition in two steps.
		First,  given $\varphi_t$ by equation (\ref{eq:drift}) we prove that $\mathbb{E}[e^{f_t} \vert \mathcal{F}_0 ]=1, \forall t\geq0$.
		This fact is a consequence of equation (\ref{eq:MartCond2}).
		
		Second, we check the martingale condition; that is,  $\forall s,t$ s.t $0\leq s\leq t$
		\[
		\mathbb{E}\left[F_t(T)\middle \vert \mathcal{F}_s \right]=F_0\left(T\right)\mathbb{E}\left[e^{f_t-f_s+f_s}\middle \vert \mathcal{F}_s \right]=e^{f_s}F_0\left(T\right)\mathbb{E}\left[e^{f_t-f_s}\right] = F_s(T) \mathbb{E}\left[e^{f_t-f_s}\right] \;\;.
		\]
		The proposition is proven once we prove that $ \mathbb{E}\left[e^{f_t-f_s}\right]=1$.
		
		This equality holds because $f_t$ is additive; that is,  process  increments are independent
		\[
		\mathbb{E}\left[e^{f_t} \vert \mathcal{F}_0 \right] = \mathbb{E}\left[e^{f_t-f_s} \vert \mathcal{F}_0 \right] \; \mathbb{E}\left[e^{f_s} \vert \mathcal{F}_0 \right] \; ,
		\]
		then
		\[\pushQED{\qed} 
		\mathbb{E}\left[e^{f_t-f_s}\right] = \mathbb{E}\left[e^{f_t-f_s} \vert \mathcal{F}_0 \right]=\frac{\mathbb{E}\left[e^{f_t} \vert \mathcal{F}_0 \right]}{\mathbb{E}\left[e^{f_s} \vert \mathcal{F}_0 \right]} =1    \qedhere  \popQED 
		\]

		\paragraph{Proof of Theorem \ref{theorem:semplified_f}} $ $
		
		We check that the power-law sub-case of ATS satisfies the two conditions of \textbf{Theorem \ref{theorem:f_Additive}}.
		
		First, we verify that $g_1(t)$, $g_2(t)$, and $g_3(t)$ are non decreasing.
		\begin{enumerate}
			
			\item      \begin{equation*}
			g_1(t)=(1/2+\bar{\eta} t^\delta)-\sqrt{\left(1/2+\bar{\eta}t^\delta\right)^2+2\bar{\sigma}^2(1-\alpha)/(\bar{k}t^\beta)} \label{equation:exp4}
			\end{equation*}
			
			is non decreasing because its derivative w.r.t. $t$ is always greater or equal than zero for any $t\geq 0$.
			
			\begin{align*}
			\frac{d}{dt}\left((1/2+\bar{\eta} t^\delta)-\sqrt{\left(1/2+\bar{\eta}t^\delta\right)^2+\frac{2(1-\alpha)t^{-\beta}}{\bar{\sigma}^2\bar{k}}}\right)&\geq 0\\
			\frac{1-\alpha}{2\bar{\sigma}^2\bar{k}}\left(\frac{\beta t^{-\delta-\beta/2}}{\bar{\eta}\delta}\right)^2-\frac{\beta t^{-\delta}}{2\bar{\eta}\delta}-\frac{\beta}{\delta}&\geq  1\;\;.
			\end{align*}
			The last inequality is verified for any $t$ if and only if $\beta\geq -\delta$. The inequality holds due to the hypotheses   $\delta\leq 0$ and  $\beta> -\delta$.
			\item  \begin{equation*}
			g_2(t)=-(1/2+\bar{\eta} t^\delta)-\sqrt{\left(1/2+\bar{\eta}t^\delta\right)^2+2\bar{\sigma}^2(1-\alpha)/(\bar{k}t^\beta)}  \label{equation:exp3}
			\end{equation*} is non decreasing for any $t\geq 0$:  it is the sum of  two non decreasing functions because of the conditions $\beta\geq 0$ and $\delta\leq 0$.
			\item  \begin{equation*}
			g_3(t)=\frac{\sqrt{\bar{\sigma}^4t^{2/\alpha-2\beta(1-\alpha)/\alpha}\left(1/2+\bar{\eta} t^\delta\right)^2+2t^{-\beta+2/\alpha-2\beta(1-\alpha)/\alpha}\bar{\sigma}^2(1-\alpha)/(\bar{k})}}{\bar{k}^{(1-\alpha)/\alpha}} \label{equation:exp1}
			\end{equation*}
			is non decreasing for any $t\geq 0$: it is the sum of three non decreasing functions of $t$ (positive powers) elevated to a positive power because of the conditions $\beta\leq\frac{1}{1-\alpha/2}$ and $\delta> \frac{\beta(1-\alpha)-1}{\alpha}$.
			
		\end{enumerate}
		Second, we verify that $t\,\sigma^2_t\,\eta_t$ and $t\,\sigma_t^{2\,\alpha}\,\eta_t^\alpha/k_t^{1-\alpha}$  go to zero. The expressions $t^{1+\delta}\,\bar{ \sigma}^2\, \bar{\eta}$ and $t^{1+\delta\alpha-\beta(1-\alpha)}\,\bar{\sigma}^{2\alpha}\,\bar{ \eta}^\alpha/\bar{k}^{1-\alpha}$ go to zero as $t$ goes to zero because of the conditions  $\delta>-\min\left(\beta,\frac{1-\beta(1-\alpha)}{\alpha}\right)$ and $\beta\leq \frac{1}{1-\alpha/2}$	\qed

		\paragraph{Proof of Proposition \ref{theorem:NewAdditive}} $ $

		We prove the thesis using the definition of additive process \citep[][Def.14.1 p.455]{Cont}.
		\begin{enumerate}
			\item By hypothesis $r_0=0$ and by definition of additive process $X_0=0$ almost surely. Thus, $X_{r_0}=0$ almost surely.
			\item Independence of increments follows from the monotonicity of $r_t$.
			\item Stochastic continuity w.r.t. time  follows from stochastic continuity of the additive process and continuity of the function $r_t$ \qed
			
		\end{enumerate}

		\section{ Parameter estimation}
		
		In physics and engineering, all measurements are subject to some uncertainties or ``errors". 
		Error analysis is a vital part of any quantitative study \citep[see, e.g.][]{Taylor1997}. In this appendix, we estimate pricing errors and ``propagate" them to model parameters. This is a crucial passage to verify the quality of the proposed model.
		
		\bigskip
		
		First, we estimate pricing errors.
		In finance, the idea of considering the bid-ask spread in market prices as a sort of measurement error of ``true'' prices is well known and 
		goes back to the seminal paper of \citet{roll1984simple}. 
		He considers the price $y=y^*+q(y_{ask}-y_{bid})/2$, where $y$ is the observed price, $y^*$ the unobserved true price, and $q$ a binomial r.v. that takes value in $\{-1,1\}$ with equal probability, where $-1$ corresponds to the bid price and $+1$ to the ask price.
		Modeling the uncertainty with such a distribution, the relation between bid-ask spread and price standard deviation $\Sigma_y$ 
		is $\Sigma_y= (y_{ask} - y_{bid})/2$. 
		More recently, \citet{george1991estimation}  propose an extended formulation of the price $y=y^*+\pi q(y_{ask}-y_{bid})/2$, where $\pi$ is the unobserved proportion of the spread due to the so-called order processing cost; 
		$\pi$ is estimated from market data as a value 0.8 and in all cases analyzed in \citet{george1991estimation} is observed a value greater than 0.5.
		Conservatively, $\pi$ can be chosen as 0.5, obtaining the relation $\Sigma_y= (y_{ask} - y_{bid})/4$.
		
		Another possibility, in the plain vanilla option market for equity indices that we consider in this study, 
		is to model the true price $y$ as a Gaussian random variable with a mean equal to the mid-market price $(y_{ask} + y_{bid})/2$ and bid and ask prices chosen as symmetric quantiles. This represents more closely what is observed in this derivative market. 
		On the one hand, it is standard for a market player to pass through an options broker to work the order. Generally, real trades are closer to the mid-market than to bid/ask prices \citep[see, e.g.][]{petersen1994posted}. 
		On the other hand, it is not sure that a market player trades within the bid-ask spread. 
		In some rare cases, a trade can take place at a price higher (lower) than the ask (bid) price: 
		it can happen because the bid-ask enlarges due to sudden movements in the underlying or in presence of a very large trade, such as the hedging of a large exposure. 
		It is rather difficult to estimate how rare these events are. 
		They can happen roughly around the $5\%$ of the cases (i.e. $y_{ask} - y_{bid} \simeq 2 \times 1.96 \; \Sigma_y$).\\
		For this reason, in this paper, we consider the measurement error in prices as Gaussian and related to the bid-ask spread via 
		$\Sigma_y = (y_{ask} - y_{bid})/4$. 
		With this choice,  the relation between prices standard deviation and the bid-ask spread is equal to the one obtained by  \citet{george1991estimation}.
		
		\bigskip
		
		Second, we ``propagate''  to model parameters this measurement error in prices.
		In applied statistics, the propagation of uncertainties is a standard technique \citep[see, e.g.][]{Taylor1997, ryan2008modern}. We briefly recall some main results present in the literature for the models (\ref{eq:linear1}), (\ref{eq:linear2}) and (\ref{eq:linear3}) considered; then we describe the calibration procedure adopted in the paper.
		
		\smallskip
		
		Consider the linear model
		\begin{equation}
		y=Zg+\epsilon\;\;, \label{eq:linear1}
		\end{equation}  where $y\in \mathbb{R}^n$ is the response vector, $Z\in \mathbb{R}^{n\times(r+1)}$ is 
		the explanatory variables matrix, $\epsilon\sim N_n\left(0,\Sigma\right)$, 
		$\Sigma\in \mathbb{R}^{n\times n}$ is the diagonal response vector variance-covariance matrix, $g\in \mathbb{R}^{r+1}$ is the unobserved coefficient vector. We indicate with $N_n\left(\mu,\Sigma\right)$ an n-dimensional Gaussian distribution with mean $\mu$ and variance $\Sigma$.
		We perform a weighted linear regression with weights $W\in \mathbb{R}^{n\times n}$, a diagonal matrix. The least square solution is 
		\[\hat{g}=\left(Z'W Z\right)^{-1}Z'W Y\;\;,\] 
		where $Y\in \mathbb{R}^n$ is the observed response vector \citep[see, e.g.][Ch.3, pp.115-116]{ryan2008modern}.
		Thus, $\hat{g}$ is the Gaussian linear combination of Gaussian random variables: 
		\begin{equation}
		\hat{g}\sim N_{r+1}\left(g,\left(Z'WZ\right)^{-1} Z'W\Sigma WZ'(Z'WZ)^{-1}\right) \;\; . 
		\label{eq:linear param}
		\end{equation}
		
		In the weighted non-linear regression case, it is possible to obtain a similar result  \citep[see, e.g.][Ch.2, pp.21-24]{seber1989nonlinear}. 
		Consider the model
		\begin{equation}
		y_i=f(g,z_i)+\epsilon_i \label{eq:linear2}
		\end{equation}
		where $y_i$ is the $i^{th}$ component of the response vector $y\in \mathbb{R}^n$, $\epsilon_i$ is the $i^{th}$ component of the error vector $\epsilon\sim N_n\left(0,\Sigma\right)$,  $z_i$ is the 
		$i^{th}$  row of the explanatory variables' matrix. 
		Similarly,
		the coefficients of a non-linear regression are:
		\begin{equation}
		\hat{g}\sim N_{r+1}\left(g,\left(F'WF\right)^{-1}F'W\Sigma W'F(F'WF)^{-1}\right)\;\;,
		\label{eqmatr}
		\end{equation}
		where $F\in \mathbb{R}^{n\times(r+1)}$  is s.t. its $(i,j)$ element is 
		\[
		F^{i,j}=\frac{\partial f}{\partial g_j}\vert_{g,z_i}\;\;
		\]
		and $g_j$ is the $j^{th}$ component of $g$.
		
		\bigskip
		
		In the literature, the case that takes into account Gaussian correlated errors on both the response vector and the explanatory variables is available for the fitting of a straight line 
		\citep[see, e.g.][]{york1968least}.
		Consider the model \begin{equation}
		y_i=a+b(z_i+\epsilon_{z_i})+\epsilon_{y_i}\;\;, \label{eq:linear3}
		\end{equation}
		with $y_i$ and $z_i$ subjected to Gaussian errors with variance $\Sigma_{z_i}$ and $\Sigma_{y_i}$ and covariance $\Sigma_{z_i,y_i}$.
		The estimated slope and intercept  $\hat{a}$ and $\hat{b}$ can be obtained through a fast iterative procedure.
		In the first order approximation 
		\begin{align}\label{eq:eqlin}\begin{split}
		\hat{a}& \sim N\left(a,\Sigma_a\right)\;\;\\
		\hat{b} &\sim N\left(b,\Sigma_b\right)\;\;,
		\end{split}
		\end{align} 
		where the expressions of $\Sigma_a$ and $\Sigma_b$ are reported in \citet[$1^{st}$ equation in p.324]{york1968least}.
		
		\bigskip
		
		In this paper,
		the calibration procedure is divided into two steps. 
		
		First, for a given maturity $T$,
		we deal with the non-linear problem  and we calibrate from market data the three time-dependent parameters $k_T$, $\sigma_T$ and $\eta_T$
		on options with different strikes.
		The distribution of the estimated parameters can be obtained using equation (\ref{eqmatr}). 
		We construct $\Sigma$ through all observed bid and ask prices at the given maturity: the diagonal value is equal to ${\left(y_{ask}- y_{bid}\right)^2}/{16}$. 
		The matrix of weights $W$, as standard in the option market, is chosen as the identity matrix because the bid-ask spread does not differ significantly in the market prices in the calibration dataset. 
		Consequently, the calibration results of different models can be easily compared as shown in Section {\bf \ref{section:calibration}}, 
		where we compare ATS with LTS and Sato models.
		As result of this step, we obtain a variance-covariance matrix $\Sigma_T\in \mathbb{R}^{3\times3}$ of the estimated parameters $(k_T,\sigma_T^2,\eta_T)$ for every maturity $T$. 
		
		Then, to estimate the scaling parameters of model (\ref{eq:scaling}), we rewrite the parameters  
		definition w.r.t. $\theta:=T \sigma^2_T$ in log-log scale as
		\begin{align*}
		\ln {\hat k}_\theta&=\ln \bar{k}+\beta \ln\theta\\
		\ln {\hat \eta}_\theta&=\ln \bar{\eta}+\delta \ln\theta \;\;.
		\end{align*}
		The estimated variance and covariance of $	\ln {\hat k}_\theta$, $\ln {\hat \eta}_\theta$ and $ \ln\theta$ are obtained by a first-order expansion
		\[\begin{array}{ll}
		\left\{
		\begin{array}{lcl}
		Var\left(	\ln {\hat k}_\theta \right)&=& \displaystyle \frac{\Sigma_{T}^{1,1}}{k_T^2} +\frac{\Sigma_{T}^{2,2}}{\sigma_T^4}+2\frac{\Sigma_{T}^{2,1}}{k_T\sigma_T^2} \\[4mm]
		Var\left( 	\ln {\hat \eta}_\theta\right)&=& \displaystyle \frac{\Sigma_{T}^{3,3}}{\eta_T^2}\\[4mm]
		Var\left(	\ln \theta \right)&=& \displaystyle \frac{\Sigma_{T}^{2,2}}{\sigma_T^4} 
		\end{array}
		\right.

		&\left\{\begin{array}{lcl}
		Cov\left(	\ln {\hat k}_\theta, 	\ln \theta\right)&=& \displaystyle \frac{\Sigma_{T}^{2,2}}{\sigma_T^4}+\frac{\Sigma_{T}^{2,1}}{k_T\sigma_T^2} \\[4mm]
		Cov\left(	\ln {\hat \eta}_\theta, \ln\theta \right)&=&\displaystyle \frac{\Sigma_{T}^{2,3}}{\eta_T\sigma_T^2} 
		\end{array}\;\;,\right.
		\end{array}
		\] 
		where T is the maturity corresponding to the $\theta$ of interest. 
		The distributions of the estimated parameters $\beta$, $\delta$, $\bar{k}$, and $\bar{ \eta}$ are the one identified in equation (\ref{eq:eqlin}). The weights selected in the minimization procedure 
		\citep[see][equation (1), p.320]{york1968least} 
		are $1/ Var (	\ln {\hat k}_\theta )$ in the regression on  $\ln {\hat k}_\theta$ and $1/	Var( 	\ln {\hat \eta}_\theta)$ in the regression on  
		$\ln \hat{ \eta}_\theta$. The weights of the explanatory variable  $\ln \theta$ are $1/Var(\ln \theta)$.
		
		\smallskip
		
		Finally, from the confidence intervals for 
		$k_T$ and $\eta_T$ we can also get the confidence intervals of the skewness and the excess kurtosis at a given maturity.
		We are able to obtain skewness and excess kurtosis of ATS thanks to the identity in law with LTS \citep[for the moments of LTS see, e.g.][p.129]{Cont}.
		The linear regression of these two higher moments, 
		w.r.t. the squared root of time, is realized by computing the Gaussian errors (\ref{eq:linear param}), in the first-order approximations, of skewness and excess kurtosis.
		
	\end{appendices}
\end{document}